\documentclass[pdflatex,sn-mathphys-num]{sn-jnl}

\usepackage{graphicx}
\usepackage{multirow}
\usepackage{amsmath,amssymb,amsfonts}
\usepackage{amsthm}
\usepackage{mathrsfs}
\usepackage[title]{appendix}
\usepackage{xcolor}
\usepackage{textcomp}
\usepackage{manyfoot}
\usepackage{booktabs}
\usepackage{algorithm}
\usepackage{algorithmicx}
\usepackage{algpseudocode}
\usepackage{listings}
\usepackage{hyperref}

\usepackage{tikz}
\usetikzlibrary{trees, calc, decorations.pathreplacing}

\usepackage{eucal}          
\usepackage{amscd}          
\usepackage{latexsym}
\usepackage{mathtools}      
\usepackage{dsfont}

\usepackage[all]{xy}

\usepackage{tcolorbox}

\usepackage{float}          
\usepackage{multicol}

\usepackage{subcaption}

\usepackage{epigraph}

\theoremstyle{thmstyleone}
\newtheorem{thm}{Theorem}[section]

\newtheorem{lem}[thm]{Lemma}

\newtheorem{prop}[thm]{Proposition}
\newtheorem{proposition}[thm]{Proposition}

\newtheorem{cor}[thm]{Corollary}

\theoremstyle{thmstyletwo}
\newtheorem{example}[thm]{Example}

\newtheorem{remark}[thm]{Remark}

\theoremstyle{thmstylethree}
\newtheorem{definition}[thm]{Definition}

\newtheorem*{theorem*}{Theorem}
\newtheorem*{corollary*}{Corollary}

\newtheorem*{Satz*}{Satz}

\newtheorem*{Proposal*}{Proposal}

\newcommand{\raiz}{\mathfrak{r}}

\DeclareMathOperator{\diam}{diam}

\begin{document}

\title[]{Spectral Geometry and Heat Kernels on Phylogenetic Trees.}

\author*[1]{\fnm{Ángel Alfredo} \sur{Morán Ledezma}}\email{angel.ledezma@kit.edu}

\affil*[1]{\orgdiv{Institute of Photogrammetry and Remote Sensing}, \orgname{Karlsruhe Institute of Technology},
  \orgaddress{\street{Englerstr. 7}, \city{Karlsruhe}, \postcode{76131},
  \state{Karlsruhe}, \country{Germany}}}

\abstract{
We develop a unified spectral framework for finite ultrametric phylogenetic 
trees, grounding the analysis of phylogenetic structure in operator theory 
and stochastic dynamics in the finite setting. For a given finite ultrametric measure space \((X,d,m)\), we introduce the ultrametric Laplacian $L_X$ as the generator of a continuous time Markov chain with transition rate \(q(x,y)=k(d(x,y))m(y)\). We establish its complete spectral theory, 
obtaining explicit closed-form eigenvalues and an eigenbasis supported on the 
clades of the tree. For phylogenetic applications, we associate to any 
ultrametric phylogenetic tree $\mathcal{T}$ a canonical operator 
$L_{\mathcal{T}}$, the ultrametric phylogenetic Laplacian, whose jump rates 
encode the temporal structure of evolutionary divergence. We show that the 
geometry and topology of the tree are explicitly encoded in the spectrum and 
eigenvectors of $L_{\mathcal{T}}$: eigenvalues aggregate branch lengths 
weighted by clade mass along ancestral paths, while eigenvectors are supported 
on the clades, with one eigenspace attached to each internal node. From this 
we derive three main contributions: a spectral reconstruction theorem with 
linear complexity $O(|X|)$; a rigorous geometric interpretation of the 
spectral gaps of $L_{\mathcal{T}}$ as detectors of distinct evolutionary 
modes, validated on an empirical primate phylogeny; an eigenmode decomposition 
of biological traits that resolves trait variance into contributions from 
individual splits of the phylogeny; and a closed-form centrality index for 
continuous-time Markov chains on ultrametric spaces, which we propose as a 
mathematically grounded measure of evolutionary distinctiveness. All results 
are exact and biologically interpretable, and are supported by numerical 
experiments on empirical primate data.
}


\maketitle

\tableofcontents

\section{Introduction}

Nature organizes complexity through hierarchy. From the nested basins of a river catchment to the genealogical structure of a language family, hierarchical structure emerges wherever a complex system evolves through successive branching events, leaving behind a record of its own history \cite{rammal1986ultrametricity}. Yet perhaps the most ancient and compelling 
instance of this principle is found in life itself. All living organisms carry 
within their DNA a signature of their evolutionary heritage, and by recognizing 
and studying the patterns of these signatures, biologists are able to reconstruct 
a common origin \cite{phylogenybook}. This is the idea behind the Tree of Life, 
an image Darwin already sketched in 1837, more than two decades before 
\textit{On the Origin of Species }\cite{darwin1859origin}, organizing the diversity 
of life through branching and ramification. 

The mathematical object that captures 
this structure is the \emph{ultrametric phylogenetic tree}: a rooted, weighted tree 
in which the distance from the root to every leaf is the same, encoding the 
hierarchical and temporal structure of evolutionary divergence.

Beyond the reconstruction of evolutionary history, phylogenetic trees have 
become indispensable tools across the life sciences. Their applications span 
from population genetics and phylogeography, where they enable inference of 
past demography and historical migration events, to epidemiology, where they 
have proven essential for tracing the spread of infectious diseases across 
hosts and geographies. In microbiology, they provide one of the most natural 
and powerful measures of diversity, while in ecology they shed light on 
community assembly, interspecific interactions, and species responses to 
environmental change \cite{lewitus2016}. In medicine, the "tree of cells" is studied in the evolution of a tumor \cite{phylogenybook}. And in conservation biology, 
the shape and branch lengths of a phylogenetic tree encode the evolutionary 
distinctiveness of species, a quantity central to modern prioritization 
frameworks \cite{Redding2014,Isaac2007}. 

Despite a plethora of applications, a unified 
mathematical framework for the spectral analysis of ultrametric phylogenetic 
trees, one that yields explicit formulas, exact reconstruction theorems, and 
biologically interpretable operators, has remained absent from the literature. 
This paper develops such a framework, encompassing spectral theory, stochastic 
dynamics, and trait analysis, with applications ranging from phylogenetic 
reconstruction to evolutionary conservation.

The Laplacian operator appears across a remarkably broad range of physical 
phenomena: from the propagation of waves and the diffusion of heat to the 
quantum mechanical description of electron motion and the oscillatory dynamics 
of fluids. That a single mathematical object governs such diverse phenomena 
reflects a deep connection between the geometry of the underlying space and 
the spectral properties of the operator defined on it. This connection is the 
central subject of spectral geometry \cite{Berard1986}. The field traces its 
origins to Chladni's eighteenth-century experiments with vibrating plates and 
Rayleigh's investigations into acoustics, and was further catalyzed by Kac's 
celebrated question: can one hear the shape of a drum? \cite{Kac1966}. This 
question was ultimately answered in the negative: Gordon, Webb, and Wolpert 
constructed two distinct planar domains that are isospectral yet 
non-isometric \cite{GWW1992}, showing that the spectrum of the Laplacian alone 
does not always suffice to reconstruct the underlying geometry.

This problem 
has since been extended to discrete structures such as graphs, where analogous 
questions of spectral reconstruction and geometric inference arise naturally. 
In this direction, ultrametric analysis has proven particularly effective: 
Bradley and Mor\'{a}n showed that graphs can be reconstructed from the spectrum 
of an associated $p$-adic Laplacian \cite{bradleymoranshape}. More broadly, 
the idea of extracting geometric information from spectral data has become 
central across many areas of science, and is today a driving force in machine 
learning and geometric deep learning, where eigenvalues, geometric priors, and 
graph kernels are essential ingredients in the analysis of structured 
data \cite{Chung1997, Bronstein2021GDL}.

Ultrametric trees occupy a privileged 
position in this landscape: their discrete and hierarchical structure makes the 
spectrum and eigenvectors fully explicit and, moreover, the geometry of the tree 
is encoded directly in the spectral data, enabling, as we show in this paper, 
a suite of analytical tools for the study of phylogenetic structure.

The spectral analysis of phylogenetic trees was pioneered by Lewitus and 
Morlon \cite{lewitus2016}, who introduced the modified graph Laplacian 
$\Delta = D - W$, where $D_{ii} = \sum_j w_{ij}$ is the degree matrix and $W$ 
is the full pairwise distance matrix between all $2n-1$ nodes of the tree. In 
that framework, the tree is analyzed as a network and the eigenvalues of $\Delta$ 
are used to construct a spectral density profile; eigengaps are employed 
heuristically to identify modes of diversification, including the separation of 
distinct evolutionary lineages within empirical phylogenies. On the operator 
side, Bendikov, Cygan, and Woess developed a general framework for isotropic 
Markov generators on ultrametric spaces \cite{BendikovCyganWoess2019,Bendikov2018}, defining 
hierarchical Laplacians via a choice function on the balls of the space, with 
pure point spectrum and compactly supported eigenfunctions. In a related but independent direction, Kozyrev developed a theory of 
ultrametric pseudodifferential operators on infinite ultrametric spaces, 
establishing diagonalization results in bases of ultrametric wavelets \cite{kozyrev2005}. 
Both frameworks are developed in the infinite setting and without 
reference to phylogenetic applications. The ultrametric analysis approach has been successfully applied in studying many other biological problems, see for example \cite{Zuniga2022,MoranLedezma2025,ABZ2014,mathPhys_p30,Khrennikov2021Ultrametric,DXKM2021} and, more recently, in the study of branching coral
growth and calcification dynamics \cite{FuquenTibata2026}. Both lines of work 
leave open the same fundamental question: is there a spectral framework for 
ultrametric phylogenetic trees that is simultaneously operator-theoretic, fully 
explicit, and biologically interpretable?

The present paper answers this question by introducing the \emph{ultrametric Laplacian} $L_X$ as the central object of study. A continuous-time Markov chain on $X$ describes the dynamics of a particle jumping between states of $X$ waiting at each state an exponentially distributed random time before moving to the next. In the phylogenetic setting, $X$ is the set of taxa. The rate at which the particle jumps from $x$ to $y$ is defined by a function $q(x,y) \geq 0$. In an ultrametric space, where the hierarchical topology is encoded in the distance function $d$, it is natural for this rate to depend on $d(x,y)$: taxa sharing a more recent common ancestor should communicate more readily than those whose common ancestor is more ancient. This motivates the choice $q(x,y) = k(d(x,y))$ for some positive kernel $k$, leading to the Markov generator \begin{equation*} L_X u(x) = \sum_{y \in X} k(d(x,y))\,(u(y) - u(x))\,m(y), \end{equation*} where $m$ is a probability measure on $X$. The framework developed here applies to ultrametric phylogenetic trees with arbitrary branching, not only to binary trees. 

We start Section \ref{sec:Ultrametric spaces and trees.} by introducing basic definitions. In particular, we introduce the concept of \emph{topological tree}, a rooted tree which encodes the branching topology of a given finite ultrametric space. In Section \ref{subsec:Ultrametric trees and Phylogenetic trees.}, we show how any phylogenetic tree has attached naturally an ultrametric space; through Lemmas \ref{lem:topologicaltree} and \ref{lem:topologicaltree2} we establish the well-known bijection between ultrametric trees and finite ultrametric spaces \cite{rammal1986ultrametricity}, creating a dictionary between phylogenetic terminology and the theory of ultrametric Laplacians.

In Section \ref{sec:spectral}, we introduce $L_X$, an operator acting on 
functions defined on the leaves of an ultrametric tree, and establish its 
complete spectral theory. The eigenvalues admit explicit closed-form 
expressions in terms of the geometry of the tree, and the eigenvectors are 
supported on the clades, with one eigenspace attached to each internal 
node, making the spectral structure fully transparent 
without numerical diagonalization. The first main result is a \emph{spectral 
reconstruction theorem} (see Theorem \ref{thm:reconstructiontheorem}): a labeled ultrametric phylogenetic tree can be 
recovered, up to realization, from the \emph{spectral encoding} $\sigma^e(X)$, 
a planar sequence of pairs $(\lambda_n, m_n)$ ordered by a breadth-first 
traversal of the tree, with linear complexity $O(|X|)$. To prove this result we use an explicit probability measure referred to as the \emph{Lebesgue measure of the tree}. This provides an optimal tool for the storage, access, and simulation of an ultrametric phylogenetic tree.

Following the strategy of \cite{lewitus2016}, we then study the spectral gaps in this setting. Given a phylogenetic tree \(\mathcal{T}\), we associate what we call the \emph{ultrametric phylogenetic Laplacian} of \(\mathcal{T}\), denoted by \(L_{\mathcal{T}}\), an ultrametric Laplacian attached to the underlying ultrametric space with jump rates depending on \(h_0-h(x\wedge y)\), where \(h_0\) is a fixed reference height and \(h(x\wedge y)\) denotes the height of the divergence event between taxa \(x\) and \(y\). Since in this case the jump rate follows
\[
F(h_0-h(x\wedge y))=\frac{d}{dt}\mathbb{P}(X_{t+h}=y|X_t=x)\big|_{t=0},
\]
we obtain a consistent picture linking the biological information of the phylogenetic tree to the random process: two species separated by a very old common ancestor (high height of their LCA) have a small jump rate; that is, a jump between phylogenetically distant taxa is a rare event. On the other hand, recently diverged taxa (small height of their LCA) have large jump rates and are better connected. This leads to a geometrical interpretation of the spectrum and, in particular, of the spectral gaps. The eigenvalues of the ultrametric Laplacian encode the hierarchical structure of the tree. Each eigenvalue
\[
\lambda(u) = \sum_{n \in \gamma_r(u)} m(n)\, l(n, \mathrm{Father}(n))
\]
aggregates the branch lengths $l(n, \mathrm{Father}(n))$ along the ancestral path $\gamma_r(u)$ from a clade $u$ to the root, weighted by the mass $m(n)$ of each intermediate ancestor. This reflects the relative species richness of the corresponding subtree. Large eigenvalues thus arise from clades whose ancestral paths traverse long branches or pass through taxonomically rich intermediate nodes, linking spectral magnitude to both evolutionary depth and clade diversity. Consequently, the spectral gaps reflect distinct evolutionary modes. To test the theory, we construct the phylogenetic Laplacian $L_{\mathcal{T}}$ on the tree of Primate genera with 109 leaves obtained from the TimeTree dataset \cite{Kumar2022}; the spectral gap separates distinct evolutionary modes, isolating Strepsirrhini from Simiiformes without supervision. Moreover, the choice of kernel \(F\) acts as a contrast function at different scales of the hierarchy: a sigmoid kernel achieves complete spectral separation between the parvorders Platyrrhini and Catarrhini, a separation not visible under the baseline kernel.

A further result concerns the decomposition of trait variance along the phylogenetic 
tree. We provide an explicit and natural orthonormal basis \(\{\psi_P\}\) for the ultrametric Laplacian, leading to an eigenmode decomposition of any function $f: X \to \mathbb{R}$ representing a biological trait. Moreover, if \(c_P=\langle f,\psi_P\rangle\), we can decompose the variance as $\mathrm{Var}_m(f) = \sum_P c_P^2$, where each coefficient 
$c_P$ measures the contrast between trait averages in the two clades generated 
by the split $P$, weighted by their relative mass, while also admitting a geometrical interpretation as the projection of $f$ onto $\psi_P$. This decomposition is in 
the spirit of the phylogenetic independent contrasts of 
Felsenstein \cite{Felsenstein1985} and the orthonormal variance decomposition 
of Ollier, Couteron, and Chessel \cite{Ollier2006}, but derived directly from 
the eigenbasis of $L_X$, which generalizes the Haar-like wavelets 
of \cite{Gorman2023}, rather than from an \emph{ad hoc} topological construction. 
Large coefficients signal splits where substantial trait divergence occurred 
between clades of comparable size, providing a natural and interpretable basis 
for phylogenetic comparison. As an illustration, we analyze this decomposition for three traits on the phylogenetic tree of Primate genera, where trait data were obtained from the PanTHERIA dataset \cite{Pantheria}. The cumulative variance profiles reveal that body mass accumulates most of its variance at low eigenvalues, longevity concentrates it in an intermediate spectral band, and litter size rises gradually across the entire spectrum. These distinct profiles illustrate that the decomposition not only summarizes variance but locates it along the eigenvalue axis, distinguishing traits whose variation is captured by a few deep splits from those distributed more uniformly across the tree. More broadly, the eigenmode decomposition connects naturally with the wider toolkit of phylogenetic comparative methods. A systematic comparison of the geometric framework developed here with established stochastic approaches to trait evolution on trees, both analytically and on empirical datasets, is deferred to a forthcoming companion study.

In Section \ref{sec:centrality}, we extend the 
random walk centrality of Noh and Rieger \cite{Nohcentrality2004} to the 
continuous-time setting, obtaining a closed-form expression for the CTMC 
centrality $C_{\mathrm{CTMC}}(i)$ on ultrametric spaces. The index admits a dual 
interpretation: dynamically, it quantifies the accessibility of a state under 
the stochastic evolution; topologically, it reflects the ramification of the 
path connecting the state to its ancestors in the ultrametric tree. As an 
application to phylogenetic conservation, a leaf with low centrality 
corresponds to a species with few close relatives, occupying a long and poorly 
branched lineage, precisely the signature of a phylogenetically distinct 
species whose evolutionary history is disproportionately large. Such species 
are central to conservation prioritization frameworks such as the EDGE of 
Existence program \cite{Redding2014,Isaac2007}. The CTMC centrality thus provides a 
mathematically grounded index of evolutionary distinctiveness, derived not 
from \emph{ad hoc} branch-length summation but from the stochastic geometry of the 
ultrametric space. Compared with existing measures of phylogenetic isolation such as evolutionary distinctiveness, $C_{\mathrm{CTMC}}$ offers several structural differences compared with existing measures: it incorporates information from the entire tree topology rather than only the root-to-leaf path; the kernel function provides an interpretable and mathematically justified mechanism to tune the trade-off between sensitivity to recent and ancient divergences; and the formulation extends naturally to non-ultrametric trees and phylogenetic networks, addressing a limitation noted in \cite{Redding2014} for most existing metrics. We close this introduction by noting that the entire framework can be extended to non-ultrametric phylogenetic trees, though in most cases at the cost of losing the explicit formulas.

\section{Ultrametric spaces and trees.}\label{sec:Ultrametric spaces and trees.}

The concepts of ultrametric spaces and ultrametric trees are central in this work. In this section we introduce the notions of ultrametric space, ultrametric tree and ultrametric phylogenetic tree. We observe the relationship between these three objects creating a dictionary that will be used for the rest of this work.  

\begin{definition}
  A metric space is a pair \((X, d)\), where \(X\) is a non-empty set and \(d : X \times X \to [0, \infty)\) is a function such that, for all \(x, y, z \in X\):
\begin{enumerate}
    \item \(d(x, y) = 0 \iff x = y\),
    \item \(d(x, y) = d(y, x)\),
    \item \(d(x, z) \leq d(x, y) + d(y, z)\).
\end{enumerate}

An ultrametric space is a metric space \((X, d)\) in which the metric \(d\) satisfies the strong triangle inequality:
\[
d(x, z) \leq \max\{ d(x, y), d(y, z) \} \quad \text{for all } x, y, z \in X.
\]
\end{definition}
\begin{example}
    Consider the set \(X = \{a, b, c, d\}\), and define \(d : X \times X \to \mathbb{R}\) as follows:
\[
d(x, y) =
\begin{cases}
0 & \text{if } x = y, \\
2 & \text{if } \{x, y\} \subset \{a, b\} \text{ or } \{x, y\} \subset \{c, d\}, \\
3 & \text{otherwise.}
\end{cases}
\]
It is easy to verify that \(d\) is an ultrametric on \(X\). The ultrametric structure can be represented by the following tree:

\begin{center}
\begin{tikzpicture}[level distance=1.2cm,
  level 1/.style={sibling distance=2.4cm},
  level 2/.style={sibling distance=1.2cm}]
  \node {$B_3 = X$}
    child {node {$B_2^1 = \{a, b\}$}
      child {node {$a$} }
      child {node {$b$} }
    }
    child {node {$B_2^2 = \{c, d\}$}
      child {node {$c$} }
      child {node {$d$} }
    };
\end{tikzpicture}
\end{center}

\end{example}

The concept of tree is intimately related with the one of an ultrametric space. The topology of any finite ultrametric space can be described by a tree as shown in the example above. In order to introduce terminology and fix notation we continue with some definitions. 

A \emph{graph} is a pair \((V, E)\), where \(V\) is a finite set of vertices and \(E \subset V \times V\) is a set of edges. A \emph{(combinatorial) tree} \(T\) is a connected, acyclic graph. A tree \(T\) having a distinguished node called the \emph{root}, denoted by \(r\), is called a \emph{rooted} tree. 

For each node \(n\in T\) we define the \emph{history} of \(n\) as the unique sequence of nodes connecting the root \(r\) with the node \(n\) including the extremes. We denote this set by \(\gamma_r(n)\).

The number \(|\gamma_r(n)|-1\) will be called the \emph{level} of \(n\). We say that \(m\in T\) is an \emph{ancestor} of \(n \in T\) if \(m\in \gamma_r(n)\). Given two vertices \(n\) and \(m\), their \emph{least common ancestor} is the unique vertex of maximal level that is an ancestor of both \(n\) and \(m\) and will be denoted by \(n\wedge m\). In general, a path connecting two different nodes  \(n,m \in T\setminus\{r\}\) will be denoted by \(\gamma(n,m)\). If \(m\in T\) is an ancestor of \(n\in T\) such that \(\gamma(n,m)=\{n,m\}\) then \(m\) is refer as the \emph{father} of \(n\) and we use the notation \(F(n)=m\). In this case, we also say that \(m\) and \(n\) are \emph{consecutive}. In this case \(n\) is said to be a \emph{child} of \(m\). A vertex with no children is called a \emph{leaf}, while a vertex that is not a leaf is called an \emph{internal node}.

The topology induced by the ultrametric can be described by rooted tree. Each internal node corresponds to a closed ball in \(X\) with respect to \(d\), and each leaf represents an element of \(X\). We refer to this tree as the topological tree associated with the ultrametric space \((X, d)\).

\begin{definition}
Let \((X, d)\) be a finite ultrametric space. The \emph{topological tree} associated with \((X, d)\) is a rooted tree \(T\) with the following properties:
\begin{enumerate}
    \item The set of leaves of \(T\) is in bijection with \(X\).
    \item Each internal node of \(T\) is associated with a ball in \(X\) of the form \(B(x, r) = \{y \in X : d(x, y) \leq r\}\) for some \(x \in X\) and \(r > 0\), and every ball in \(X\) arises in this way.
    \item For any two leaves \(x, y \in X\), the smallest ball containing both corresponds to their lowest common ancestor in \(T\).
\end{enumerate}
\end{definition}
Therefore there is a one to one correspondence between the balls generated by the ultrametric $d$ and the nodes $n\in T$. Henceforth, we will use both terms interchangeably when the tree $T$ refers to the topological tree. A ball will be refer as \(n\), \(B\) or \(B_n\) when necessary. We also introduce the notation \(B_n^{+}=B_{F(n)}\), and write \(n^{+}\) for \(F(n)\) when the context is clear. In particular \(B_r=X\). \newline

\subsection{Ultrametric trees and Phylogenetic trees.}\label{subsec:Ultrametric trees and Phylogenetic trees.}

For a given tree \(T\), define a \emph{branch} as the edge connecting two consecutive internal nodes \(u,v\in T\). We can associate a length function to the set of branches denoted by \(l(u, v)>0\). In general, when a tree has attached a function to its edges we call such tree a \emph{weighted tree}. The pair \(\mathcal{ T}=(T,l)\) is called an \emph{ultrametric tree} if the sum of the lengths of the branches connecting the root and any leaf is constant. 

A species is one of the most fundamental units of biology. Over time, species are shaped by evolutionary forces such as mutation and natural selection, which drive changes at both the molecular and morphological level. A key outcome of these processes is that a species may either give rise to two distinct lineages that evolve independently (a speciation event) or disappear entirely through extinction. The cumulative result of such events is a branching pattern that can be represented as a tree structure, commonly referred to as the tree of life \cite{phylogenybook}.   

A \emph{phylogenetic tree} is a tree equipped with a length function \(\mathcal{T}=(T,l)\) that represents the evolutionary history of a set of entities; it describes a hypothetical pattern of speciation events that occurred in the past, each internal node represents the speciation events, and the leaves represent the analyzed "present time" species. The length between two consecutive internal nodes represents the time between those two speciation events. A \emph{rooted phylogenetic tree} is a phylogenetic tree \(\mathcal{T}=(T,l)\) such that \(T\) is a rooted tree. If \(\mathcal{T}\) is an ultrametric tree, then \(\mathcal{T}\) is called a \emph{ultrametric phylogenetic tree}. 

Let \(\phi:X\rightarrow t(X)\) be a \emph{labeling}, that is a bijective functions from the set of leaves to the set of taxa \(t(X)\). An ultrametric phylogenetic tree \(\mathcal{T}\) equipped with a labeling will be refereed as \textit{labeled phylogenetic ultrametric tree}. An internal node of a labeled phylogenetic tree will be called a \emph{clade} of \(\mathcal{T}\).
\\
Any of these trees has a natural associated ultrametric, that is if \(\mathcal{T}\) is an ultrametric tree then for any pair of leaves \(x,y\in X\) we define
\begin{equation}
    \label{phylometric}
    d_l(x,y):=2\sum_{\substack{u\rightarrow v \\ u,v\in\gamma(x\wedge y,x)}} l(u, v).
\end{equation}
That is, the distance between two leaves is the sum of the lengths of the branches along the path connecting them which necessarily passes through  their least common ancestor. We also define the height of an internal node \(n\) as 
\[h(n):=\frac{d_l(x,y)}{2},\]
where \(n=x\wedge y\) for some \(x,y \in X\), this is well defined by the ultrametric inequality. The distance may not have a direct biological interpretation, nevertheless the height \(h\) of an internal node in a phylogenetic tree represent the time between the present and the split event. The underlying tree \(T\) of a given phylogenetic tree \(\mathcal{T}=(T,l)\) need not be binary; that is, a single node may give rise to more than two descendant lineages simultaneously.

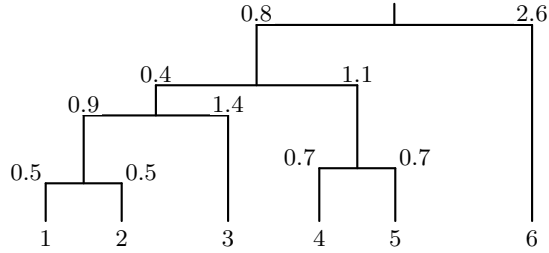
\begin{figure}[H]
  \centering
  \begin{tikzpicture}[
  thick, line cap=round,
  leaf lab/.style={font=\small},
  num/.style={fill=white, inner sep=1pt, font=\scriptsize}
]

\def\Hbase{0.0}
\def\HoneTwo{0.5}      
\def\HfourFive{0.7}    
\def\HoneTwoThree{1.4} 
\def\Hall{1.8}         
\def\Htop{2.6}         
\def\xone{0.6}
\def\xtwo{1.6}
\def\xthree{3.0}
\def\xfour{4.2}
\def\xfive{5.2}
\def\xsix{7.0}

\draw (\xone,\Hbase)   -- (\xone,\HoneTwo);
\draw (\xtwo,\Hbase)   -- (\xtwo,\HoneTwo);
\draw (\xthree,\Hbase) -- (\xthree,\HoneTwoThree);
\draw (\xfour,\Hbase)  -- (\xfour,\HfourFive);
\draw (\xfive,\Hbase)  -- (\xfive,\HfourFive);
\draw (\xsix,\Hbase)   -- (\xsix,\Htop);

\draw (\xone,\HoneTwo) -- (\xtwo,\HoneTwo);      
\draw (\xfour,\HfourFive) -- (\xfive,\HfourFive);         
\coordinate (m12) at ($(\xone,\HoneTwo)!0.5!(\xtwo,\HoneTwo)$);
\coordinate (m45) at ($(\xfour,\HfourFive)!0.5!(\xfive,\HfourFive)$);
\draw (m12) -- ($(m12)+(0,\HoneTwoThree-\HoneTwo)$);
\draw (m45) -- ($(m45)+(0,\Hall-\HfourFive)$);

\draw ($(m12)+(0,\HoneTwoThree-\HoneTwo)$) -- (\xthree,\HoneTwoThree);
\coordinate (m123) at ($($(m12)+(0,\HoneTwoThree-\HoneTwo)$)!0.5!(\xthree,\HoneTwoThree)$);

\draw (m123) -- ($(m123)+(0,\Hall-\HoneTwoThree)$);

\draw ($(m123)+(0,\Hall-\HoneTwoThree)$) -- ($(m45)+(0,\Hall-\HfourFive)$);
\coordinate (mall) at ($($(m123)+(0,\Hall-\HoneTwoThree)$)!0.5!($(m45)+(0,\Hall-\HfourFive)$)$);

\draw (mall) -- ($(mall)+(0,\Htop-\Hall)$);
\draw ($(mall)+(0,\Htop-\Hall)$) -- (\xsix,\Htop);

\coordinate (midTop) at ($($(mall)+(0,\Htop-\Hall)$)!0.5!(\xsix,\Htop)$);
\draw (midTop) -- ++(0,0.28);

\node[num, above left]  at (\xone,\HoneTwo) {0.5};
\node[num, above right] at (\xtwo,\HoneTwo) {0.5};

\node[num, above left]  at (\xfour,\HfourFive) {0.7};
\node[num, above right] at (\xfive,\HfourFive) {0.7};

\node[num, above] at ($(m12)+(0,\HoneTwoThree-\HoneTwo)$) {0.9}; 
\node[num, above] at (\xthree,\HoneTwoThree) {1.4};              
\node[num, above] at ($(m123)+(0,\Hall-\HoneTwoThree)$) {0.4};    
\node[num, above] at ($(m45)+(0,\Hall-\HfourFive)$) {1.1};        
\node[num, above] at ($(mall)+(0,\Htop-\Hall)$) {0.8};            
\node[num, above] at (\xsix,\Htop) {2.6};

\node[leaf lab, below] at (\xone,\Hbase) {1};
\node[leaf lab, below] at (\xtwo,\Hbase) {2};
\node[leaf lab, below] at (\xthree,\Hbase) {3};
\node[leaf lab, below] at (\xfour,\Hbase) {4};
\node[leaf lab, below] at (\xfive,\Hbase) {5};
\node[leaf lab, below] at (\xsix,\Hbase) {6};

\end{tikzpicture}
  \caption{Phylogenetic tree. Here \(d(1,2)=1\) and \(d(3,5)=3.6\)  }
\label{fig:phylotree}
\end{figure}
By construction, we conclude that every ultrametric (phylogenetic) tree has attached a natural ultrametric space \((X_{\mathcal{T}},d)\), where \(X_{\mathcal{T}}\) is the set of leaves of \(\mathcal{T}\). Moreover, the following lemma follows from the definitions above.
\begin{lem} \label{lem:topologicaltree}
      The topological tree \(T^{*}\) of the ultrametric space  \((X_{\mathcal{T}},d)\) is the underlying tree of  \(\mathcal{T}\), that is \(\mathcal{T}=(T^{*},l)\).  
\end{lem}
This lemma serves as a dictionary between the theory of ultrametric spaces and ultrametric (phylogenetic) trees. In other words, the \emph{skeleton} of the ultrametric (phylogenetic) tree \(\mathcal{T}\) is the topological tree of the ultrametric space attached to it. 

Given a finite ultrametric space \((X,d)\) with topological tree \(T\), it is possible to construct an ultrametric tree from it. For a given internal node \(n\in T\) , define \(h(n):=d(x,y)/2\) where \(n=x\wedge y\) for some leaves \(x,y\in X\). Define \(l(n,m):=|h(n)-h(m)|\). Then the tree \((T,l)\) is ultrametric, thus the following result follows. 
\begin{lem}\label{lem:topologicaltree2}
    Let \((X,d)\) be a finite ultrametric space, then there is an ultrametric tree \(\mathcal{T}=(T,l)\) such that \((X,d)=(X_{\mathcal{T}},d_l)\)
\end{lem}
From these two lemmas we conclude that there is a bijection between ultrametric trees and finite ultrametric spaces. Henceforth, we use the term ultrametric tree to refer to the mathematical object \(\mathcal{T}=(T,l)\) and we use the name ultrametric phylogenetic tree when we want to refer to the biological interpretations arising from this structure.

\section{Spectral geometry for phylogenetic trees}\label{sec:spectral}

An ultrametric space equipped with a measure \(m\) is called a \emph{measurable} ultrametric space. In particular in the finite setting this is just a positive function on the leaves. A \emph{probability measure} is a measure satisfying \[\sum_{x\in X}m(x)=1.\]
We will develop the theory for any \(m\).

Given a finite ultrametric space \((X, d,m)\) with topological tree \(T\) and measure \(m\), consider the operator
\[
L_X u(x) = \sum_{y\in X} k(d(x, y)) \big(u(y) - u(x)\big) \, m(y),
\]
where \(k : [0, \infty) \to \mathbb{R}\) is a given function. Although $L_X$ will be interpreted in Section\ref{sec:centrality} as the generator 
of a continuous-time Markov chain with transition rates 
$k(d(x,y))m(y)$, we first develop its spectral theory in full 
generality, as this provides the analytical foundation for the 
probabilistic results that follow.

This operator acts on functions \(u : X \to \mathbb{R}\). We refer to $L_{X}$ as the \emph{ultrametric Laplacian operator} attached to the triple $(X,d,m)$.  

Consider the (canonical) basis \(\{e_y\}_{y \in X}\) of functions on \(X\) equipped with the canonical inner product, where \(e_y(x) = \delta_{x,y}\) and \(\delta_{x,y}\) is the Kronecker delta, defined by
\[
\delta_{x,y} =
\begin{cases}
1 & \text{if } x = y, \\
0 & \text{otherwise}.
\end{cases}
\]
With respect to this basis, the operator \(L_{X}\) can be represented as a matrix, called the \emph{ultrametric Laplacian matrix}.

\begin{definition}
The \emph{ultrametric Laplacian matrix} associated with the operator \(L_X\) is the \(|X| \times |X|\) matrix \(L= (L_{x,y})_{x, y \in X}\) whose entries are given by
\[
L_{x,y} =
\begin{cases}
k(d(x, y))\, m(y) & \text{if } y \neq x, \\
- \sum_{z \neq x} k(d(x, z))\, m(z) & \text{if } y = x.
\end{cases}
\]
\end{definition}
This matrix corresponds to the matrix representation of \(L_X\) with respect to the canonical basis \(\{e_y\}_{y \in X}\).

When the kernel satisfies \(k\big(diam(B_n)\big)-k\big(diam(B_{F(n)})\big)>0 \) for any \(n\in T\setminus X \) different from the root, the ultrametric Laplacian \(L_X\) coincides with the Hierarchical Laplacian (see \cite{Angulo2019, Bendikov2018}). In order to see this connection we introduce this operator in the setting of finite ultrametric spaces. \\
For a ball \(B\), define the operator 
\[
(P_B f)(x):=\mathbf{1}_{\{x\in B\}}\;\frac{1}{m(B)}\!\sum_{y\in X} f(y)\,m(y).
\]
\begin{definition}
Let $\{a(B_n)\}_{n\in T\setminus X}\subset [0,\infty)$ be weights on balls. The hierarchical Laplacian associated with \(m\) and \(a\) is defined as
\[
 (L_a f)(x)\;=\;-\sum_{n\in T\setminus X: x\in B_n} a(B_n)\,\Big(f(x)-(P_{B_n} f)(x)\Big).
\]
\end{definition}
\begin{proposition}[Hierarchical decomposition of $L_X$]
Assume that for each $n\in T\setminus  X$  different from the root the kernel \(k\) satisfies \(k\big(diam(B_n)\big)-k\big(diam(B_{F(n)})\big)>0 \). For \(n\in T\setminus \ 
X\) different from the root define
\[
a(B_n)\;:=\;m(B_n)\Big(k\big(diam(B_n)\big)-k\big(diam(B_{F(n)}\big)\Big)\;\;\ge 0,
\]
and \[a(X):=m(B_n)k(X)\]

Then
\[
\quad L_X f \;=\; L_a f\quad
\]
\end{proposition}
\begin{proof}
    Let \(\delta_z\) the indicator function of the point \(z\in X\). Let \(y_0\in X\) with \(y_0 \neq z\). It is clear that 
    \[P_{B_n}=\delta_{z\in B_n}\frac{m(z)}{m(B_n)}, \]
    where \(\delta_{z\in B_n}\) is \(1\) if \(z\in B_n\) and zero otherwise. Therefore, we have 
    \begin{equation*} L_a\delta_z(y_0)=-\sum_{n\in\gamma_r(y_0)\setminus\{y_0\}}a(B_n)\left(0-\delta_{z\in B_n}\frac{m(z)}{m(B_n)}\right)=\sum_{n \in \gamma_r(y_0\wedge z)} a(B_n)\frac{m(z)}{m(B_n)}
    \end{equation*}
Note that 
\[
\sum_{n \in \gamma_r(y_0\wedge z)} \frac{a(B_n)}{m(B_n)}
=\sum_{n \in \gamma_r(y_0\wedge z)}\left(k(d(B_n))-k(d(B_{F(n)}))\right)
= k\left(d(y_0,z)\right).
\]
Therefore, \(L_a\delta_z(y_0)=L_X\delta_z(y_0)\). On the other hand, 
\begin{equation*}
    \begin{split}
        -\sum_{n\in\gamma_r(z)\setminus \{z\}}a(B_n)\left( 1-\frac{m(z)}{m(B_n)} \right)&=-\sum_{n\in\gamma_r(z)\setminus \{z\}}\frac{a(B_n)}{m(B_n)}(m(B_n)-m(z))\\&=-\sum_{n\in\gamma_r(z)\setminus \{z\}}\frac{a(B_n)}{m(B_n)}\sum_{y\in B_n:z\neq y}m(y)\\
        &=-\sum_{z\neq y}k(d(z,y))m(y)
    \end{split}
\end{equation*}
Therefore \(L_X=L_a\).
\end{proof}
\begin{remark}
When the kernel satisfies the strict monotonicity condition, 
$L_X$ coincides with the hierarchical Laplacian $L_C$ of 
Bendikov et al. \cite{BendikovCyganWoess2019}, restricted 
to the finite setting. Outside this regime, $L_X$ is strictly 
more general. Note that the framework of Kozyrev \cite{kozyrev2005}, 
while formally related, is developed exclusively for infinite 
ultrametric spaces and does not directly apply here.
\end{remark}

\subsection{The spectrum}

Inspired by the theory of spectral geometry and manifold learning, one of the goals of this section is to show the explicit connection between the topology of a finite ultrametric space and the eigenvalues and eigen-vectors of the operator \(L_X\). We are particularly interested in seeing how these connection leads to alternative ways to analyze a phylogenetic tree. Therefore, the description of the spectral nature of this operator is central for the theory. 

We initiate this study by constructing an orthonormal basis of eigenvectors of \(L_X\). Following \cite{Bendikov2018}, given \(n\in T\setminus X\) an internal node, the functions 
\begin{equation}
    \label{eq:eigenvectorsultrametric}\varphi_{B_n,l}=\frac{\mathbf{1}_{B_l}}{m(B_l)}-\frac{\mathbf{1}_{B_{n}}}{m(B_{n})},
\quad\text{for each child }l\in C(n).
\end{equation}
are eigenfunctions of the hierarchical Laplacian \(L_a\), and therefore, in the special case when \(k(diam(B_n))-k(diam(B_{F(n)})>0\), for any \(n\in T\setminus\{X\}\), they are also eigenfunctions of \(L_X\). 

Moreover, the function \(\varphi_0=\textbf{1}\equiv 1\), denoting the trivial eigenvector with eigenvalue \(0\),  together with the set of all these functions form a complete system of the space \(L^2(X,m)\) \cite{Bendikov2018}. In particular, notice that for a node \(n\) the set of functions \(\varphi_{B_n,l}\), with \(l\in C(n)\) span the set 
\[\mathcal{V}_n:=\left\{ \psi: \psi|_{B_l,} \, \text{is constant for all \(l\in C(n)\)},\, \sum_{l\in C(n)} m(B_l)\psi|_{B_l,}=0, \, Supp \, \psi \subset B_n\right\}.\]The dimension of this space is \(|C(n)|-1
\). This results extends to any ultrametric Laplacian \(L_X\) with positive kernel \(k\).

\begin{prop} \label{prop:eigenvalueofultrametricoperator}
    Let \(\varphi_{B_n,l}\) defined as in equation \ref{eq:eigenvectorsultrametric}. Then \(\varphi_{B_n,l}\) is an eigenfunction of \(L_X\), with eigenvalue 
    \begin{equation} \label{eigenvaluereal}
        \begin{split}
            \lambda_n&=-\sum_{l\in \gamma_r(n)} m(B_l)\left[k(diam(B_l))-k(diam(B_{F(l)})\right]\\
            &=-\sum_{y\in X\setminus B_n} k(d(x_0,y))dm(y) - k(diam(B_n))m(B_n),
        \end{split}
    \end{equation}
    where \(x_0\in B_n\), with multiplicity \(m_n:=|C(n)|-1\).
\end{prop}
\begin{proof}
    Let \(n\in T\setminus X\) an internal node. Let \[
\varphi_{B_n,l}=\frac{\mathbf{1}_{B_l}}{m(B_l)}-\frac{\mathbf{1}_{B_{n}}}{m(B_{n})},
\quad\text{for a given child }l\in C(n).
\] define the real number
\(\lambda_n:=-\sum_{y\in X\setminus B_n} k(d(x_0,y))dm(y) - k(diam(B_n))m(B_n).\) Define \(a=\frac{1}{m(B_l)}-\frac{1}{m(B_{n})}\) and \(b=-\frac{1}{m(B_{n})}\). Let \(x\in B_l\), therefore, \(\varphi_{B_n,l}(x)=a\), and 
\[\sum_{y\in B_l}k(d(x,y))(\varphi_{B_n,l}(y)-\varphi_{B_n,l}(x))m(y)=0.\] Therefore
\[L_X \,\varphi_{B_n,l}(x)=\sum_{y\in B_{n}\setminus B_l}k(d(x,y))(b-a)m(y)+\sum_{y\in X\setminus B_n}k(d(x,y))(0-a)m(y).\]
Since \(x\in B_l\), for any \(y\in B_n\setminus B_l\), \(k(d(x,y))=k(diam(B_n))\) and 
\begin{equation*}
    \begin{split}
        L_X \,\varphi_{B_n,l}(x)&=k(diam(B_n))(b-a)m(B_n\setminus B_l)-a\sum_{y\in X\setminus B_n}k(d(x,y))m(y)\\
        &=a\left[-k(diam(B_n))\left(1-\frac{b}{a}\right)m(B_n\setminus B_l)+\sum_{y\in X\setminus B_n}k(d(x,y))m(y).\right]
    \end{split}
\end{equation*}
In a similar way, for \(x\in B_n\setminus B_l\) where \(\varphi_{B_n,l}(x)=b\), the following holds
\[L_X \,\varphi_{B_n,l}(x)=b\left[-k(diam(B_n))\left(1-\frac{a}{b}\right)m( B_l)+\sum_{y\in X\setminus B_n}k(d(x,y))m(y).\right]\]
A direct computation leads to the following equalities
\[\left(1-\frac{b}{a}\right)m(B_n\setminus B_l)=m(B_n),\] and 
\[\left(1-\frac{a}{b}\right)m(B_l)=m(B_n).\]
Hence, in both cases 
\[L_X \, \varphi_{B_l,n}(x)=\varphi_{B_l,n}(x) \lambda_n.\]
Lastly, if \(x\in X\setminus B_n\) then 
\[L_X \, \varphi_{B_l,n}(x)=\sum_{y\in X}k(d(x,y)) \varphi_{B_l,n}(y)m(y)=0,\]
since \(\varphi_{B_l,n}\) has mean zero.
\end{proof}
Extracting geometric information about an object in terms of the spectra of an operator is the core idea behind spectral geometry. It turns out that for a certain measure \(\mu\) the spectrum of the operator \(L_X\) allow us to recover the diameters of the balls of the ultrametric space \((X,d)\), when the kernel \(k\) is a bijection. Moreover, after a certain ordering and padding all the geometrical and topological information can be recovered from the modified sequence. In order to prove this assertion we need to define a \emph{decoration} of a tree and the \textit{Lebesgue measure} of a tree.  
\begin{definition}
Let \(T\) be the topological tree associated with a finite ultrametric space \((X, d)\), with root \(r\).
  \begin{enumerate}
      \item A decoration of  \(T\) is a bijection \(\varphi:T \rightarrow A\), where \(A\) is a set. 
      \item  For each node \(v\) in \(T\), denote by \(C(v)\) the set of its children. The \emph{Lebesgue measure} \(\mu\) on the tree is defined recursively as follows:

\begin{enumerate}
    \item \(\mu(r) = 1\).
    \item For any node \(v\) with children \(C(v) = \{v_1, \dots, v_k\}\), set \(\mu(v_i) = \mu(v) / |C(v)|\) for each \(i = 1, \dots, k\), where \(|C(v)|\) denotes the cardinality of \(C(v)\).
    \item For any leaf \(x \in X\), \(\mu(x)\) is determined by the product of the inverses of the cardinalities along the unique path from \(r\) to \(x\):
    \[
    \mu(x) = \prod_{v \in \gamma_r(x)\setminus\{x\}} \frac{1}{|C(v)|}.\]
\end{enumerate}
  \end{enumerate}
\end{definition}
For example, an ultrametric tree \(\mathcal{T}=(T,l)\) can be viewed as the tree \(T\) decorated with the diameters of the balls attached to its ultrametric space via the function \(l\),  adopting the convention that each leaf \(x\in X\) is decorated with a zero. The spectrum leads to another decoration of the form \((\lambda_n, m_n)\), here we label the leafs by the pairs \((0,0)\).

\begin{proposition}
  \label{decorationequiv}
  Let \((X,d,\mu)\) be finite ultrametric space equipped with its Lebesgue measure and \(T\) its topological tree. Then the decoration \( ( \lambda_n,m_n)_{n \in T\setminus X} \) consisting of eigenvalues with its multiplicities of the operator \(L_X\) and \((0,0)\) in the leafs  reconstructs the decoration \((k(\diam(B_n)))_{n\in T\setminus X}\) and \((0,0)\) in the leafs.
\end{proposition}

\begin{proof}
  Let \(r\) be the root of \(T\). Then \(\lambda_r=-k(\diam(B_r))\). Since \(|C(r)|=m_r+1\), we know that \(B_r\)  decomposes into \(m_r+1\) balls of measure \(\frac{1}{m_r+1}\). Let \(n\in C(r)\). Then 
  \[\lambda_n=-k(\diam(B_r))\frac{m_r}{m_r+1}-\frac{1}{m_r+1}k(\diam(B_n)).\]
  Therefore \(\diam(B_n)\) is completely determined by \(\lambda_n\). We now proceed by induction on the levels of \(T\). Assume that for levels \(1,...,l-1\), we have decorated the nodes. Then for \(n\) at level \(l\), we have

  \[
\lambda_n = \sum_{j \in \gamma_r(n)\setminus \{n\}} k(\diam(B_j))\frac{\mu(B_j)}{m_{j}+1}  
- \frac{1}{m_{F(n)}+1}\, k(\diam(B_n))
\]
Therefore \(k(diam(B_n))\) is determined by \(\lambda_n\). This completes the proof.
\end{proof}

\begin{definition}[Spectral encoding]
Let $T$ be a rooted tree with each node $n$ decorated by a pair $(\lambda_n, m_n)$, with $(0,0)$ assigned to all leaves. The canonical ordering of these pairs is defined inductively as follows:

\begin{enumerate}
    \item \textbf{Base step:} Start with the root node $r$. The first element of the sequence is $(\lambda_r, m_r)$.
    \item \textbf{Inductive step:} Assume that the sequence has been constructed up to some level, and let $S$ be the list of pairs added at the previous step. For each pair in $S$ that is not $(0,0)$ (i.e., for each internal node), add to the sequence, in the order determined by the tree, that is, visiting the nodes in breadth-first traversal order, the pairs corresponding to its children. For each leaf child, add $(0,0)$ in the respective position.
    \item \textbf{Termination:} Repeat the inductive step until no new internal nodes remain to expand; that is, until all subsequent pairs correspond to leaves and are $(0,0)$.
\end{enumerate}

\end{definition}

The spectral encoding  will be denoted by \(\sigma^{e}(X)\). Two labeled ultrametric phylogenetic trees
\[
\mathcal{T}_1=(T_1,l_1,\phi_1)
\quad \text{and} \quad
\mathcal{T}_2=(T_2,l_2,\phi_2)
\]
are said to be \textit{isomorphic} if there exists a bijection
\[
\Psi:V(T_1)\rightarrow V(T_2)
\]
such that
\[
(u,v)\in E(T_1) \iff (\Psi(u),\Psi(v))\in E(T_2),
\qquad
l_1(u,v)=l_2(\Psi(u),\Psi(v)),
\]
and, for every leaf \(x\in X_1\),
\[
\phi_1(x)=\phi_2(\Psi(x)).
\]
Two labeled ultrametric phylogenetic trees are called
\textit{equivalent realizations} if they are isomorphic. Example of two realizations of a labeled phylogenetic tree are shown below in Figure \ref{fig:ultrametric_equivalent}.

\begin{figure}[H]
\centering
\begin{tikzpicture}[
  thick, line cap=round,
  leaf lab/.style={font=\small}
]

\def\HbaseL{0.0}
\def\HbcL{0.9}
\def\HallL{1.8}

\def\xaL{0.8}
\def\xbL{1.9}
\def\xcL{3.0}

\draw (\xaL,\HbaseL) -- (\xaL,\HallL);
\draw (\xbL,\HbaseL) -- (\xbL,\HbcL);
\draw (\xcL,\HbaseL) -- (\xcL,\HbcL);

\draw (\xbL,\HbcL) -- (\xcL,\HbcL);
\coordinate (mbcL) at ($(\xbL,\HbcL)!0.5!(\xcL,\HbcL)$);

\draw (mbcL) -- ($(mbcL)+(0,\HallL-\HbcL)$);

\draw (\xaL,\HallL) -- ($(mbcL)+(0,\HallL-\HbcL)$);

\node[leaf lab, below] at (\xaL,\HbaseL) {$a$};
\node[leaf lab, below] at (\xbL,\HbaseL) {$b$};
\node[leaf lab, below] at (\xcL,\HbaseL) {$c$};

\node at (4.7,0.95) {$\equiv$};

\def\HbaseR{0.0}
\def\HbcR{0.9}
\def\HallR{1.8}

\def\xbR{6.0}
\def\xcR{7.1}
\def\xaR{8.2}

\draw (\xbR,\HbaseR) -- (\xbR,\HbcR);
\draw (\xcR,\HbaseR) -- (\xcR,\HbcR);
\draw (\xaR,\HbaseR) -- (\xaR,\HallR);

\draw (\xbR,\HbcR) -- (\xcR,\HbcR);
\coordinate (mbcR) at ($(\xbR,\HbcR)!0.5!(\xcR,\HbcR)$);

\draw (mbcR) -- ($(mbcR)+(0,\HallR-\HbcR)$);

\draw ($(mbcR)+(0,\HallR-\HbcR)$) -- (\xaR,\HallR);

\node[leaf lab, below] at (\xbR,\HbaseR) {$b$};
\node[leaf lab, below] at (\xcR,\HbaseR) {$c$};
\node[leaf lab, below] at (\xaR,\HbaseR) {$a$};

\end{tikzpicture}

\caption{Equivalent realizations of the same labeled ultrametric phylogenetic tree.}
\label{fig:ultrametric_equivalent}
\end{figure}
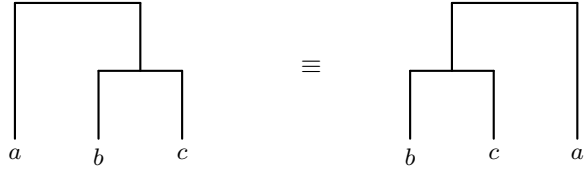
Therefore, it is clear that the spectral encoding determines a realization of the labeled
ultrametric phylogenetic tree. Different realizations corresponding to permutations of children at internal nodes produce different spectral encodings but represent equivalent labeled trees.

\begin{example} \label{examplecanonicalsequence}
Consider the following decorated rooted tree:
\begin{center}
\begin{tikzpicture}[level distance=1.2cm,
  every node/.style={draw=none},
  level 1/.style={sibling distance=3.5cm},
  level 2/.style={sibling distance=1cm}]
  \node {$r$}
    child {node {$n_1$}
      child {node {$a$}}
      child {node {$b$}}
      child {node {$c$}}
    }
    child {node {$n_2$}
      child {node {$d$}}
      child {node {$n_3$}
        child {node {$e$}}
        child {node {$f$}}
      }
    };
\end{tikzpicture}
\end{center}

The canonical ordered sequence of pairs is:
\[
(\lambda_r, m_r),\ 
(\lambda_{n_1}, m_{n_1}),\ 
(\lambda_{n_2}, m_{n_2}),\ 
(0,0),\ (0,0),\ (0,0),\ (0,0),\ (\lambda_{n_3}, m_{n_3}),\ 
(0,0),\ (0,0)
\]
where each $(0,0)$ corresponds to a leaf, and each $(\lambda_v, m_v)$ corresponds to an internal node.

\end{example}

\begin{thm}[Spectral reconstruction theorem for ultrametric trees.]\label{thm:reconstructiontheorem}
 Let \((\mathcal{T},\phi)\) be a labeled phylogenetic ultrametric tree.
Then \((\mathcal{T},\phi)\) can be reconstructed up to realization from
the spectral encoding of the ultrametric Laplacian \(L_X\) associated
with the underlying ultrametric space \((X,d,\mu)\), where \(\mu\) is
the Lebesgue measure.
\end{thm}

\begin{proof}
  Let \((\mathcal{T},\phi)\) be a labeled phylogenetic ultrametric tree.
  It is clear that the spectral encoding completely reconstruct a realization of the topological tree \(T\) decorated with the pairs \((\lambda_n,m_n)\). By Proposition \ref{decorationequiv} this decoration is equivalent to the diamenter decoration \(\diam(B_n)\). Since the finite ultrametric space \((X,d)\) is characterized by the topological tree \(T\) decorated in this way, the result follows. 
\end{proof}

\begin{remark}
    Notice that in the case of a binary tree, the decoration \((\lambda_n,m_n)\) can be substituted by the decoration \(\lambda_n\) since in this case \(m_n=1\) for all internal node \(n\).
\end{remark}

\begin{example} (Counterexample, isospectral spaces.) In this example we show the existence of two ultrametric spaces which are not isomorphic but have the same eigenvalues. 
Consider the following decorated rooted tree:
\begin{center}
\begin{tikzpicture}[level distance=1.2cm,
  every node/.style={draw=none},
  level 1/.style={sibling distance=3.5cm},
  level 2/.style={sibling distance=1.3cm},
  level 3/.style={sibling distance=0.8cm}]
  \node {$\hat{r}$}
    child {node {$\hat{n}_1$}
      child {node {$\hat{a}$}}
      child {node {$\hat{b}$}}
      child {node {$\hat{n}_3$}
        child {node {$\hat{c}$}}
        child {node {$\hat{d}$}}
      }
    }
    child {node {$\hat{n}_2$}
      child {node {$\hat{e}$}}
      child {node {$\hat{f}$}}
    };
\end{tikzpicture}
\end{center}

The canonical ordered sequence of pairs is:
\[
(\lambda_{\hat{r}}, m_{\hat{r}}),\ 
(\lambda_{\hat{n}_1}, m_{\hat{n}_1}),\ 
(\lambda_{\hat{n}_2}, m_{\hat{n}_2}),\ 
(0,0),\ (0,0),\ (\lambda_{\hat{n}_3}, m_{\hat{n}_3}),\ 
(0,0),\ (0,0),\ (0,0),\ (0,0).
\]
We will show that we can construct with this tree an isospectral ultrametric space to the one given in Example \ref{examplecanonicalsequence}. Assume that \(\diam(B_{n_i})=\diam(B_{\hat{n}_i})\), for \(i=1,2\) and \(diam(B_r)=\diam(B_{\hat{r}})\). Therefore it is easy to see that the first three pairs of eigenvalues are equal. Assume now that \(k(x,y)=d(x,y)\). Then imposing \(\lambda_{n_3}=\lambda_{\hat{n}_3}\), we obtain the equation 
\[\diam(B_{n_3})=\frac{1}{3}\diam(B_{n_2})+\frac{2}{3}\diam(B_{\hat{n}_3}).\]
Therefore, if we define the diameter of \(B_{n_3}\) in this way, we can construct two non-isometric but isospectral spaces. 
\end{example}
\begin{example}
    Consider the phylogenetic tree of Figure \ref{fig:phylotree}. Consider the gravitational kernel \(k(r)=\frac{1}{r^{2}}\) and take \(m(y)=1/|X|\), the counting measure. The respective eigenvalues are shown in the figure below: 

    \begin{figure}[H]
        \centering
        \includegraphics[width=0.7\textwidth]{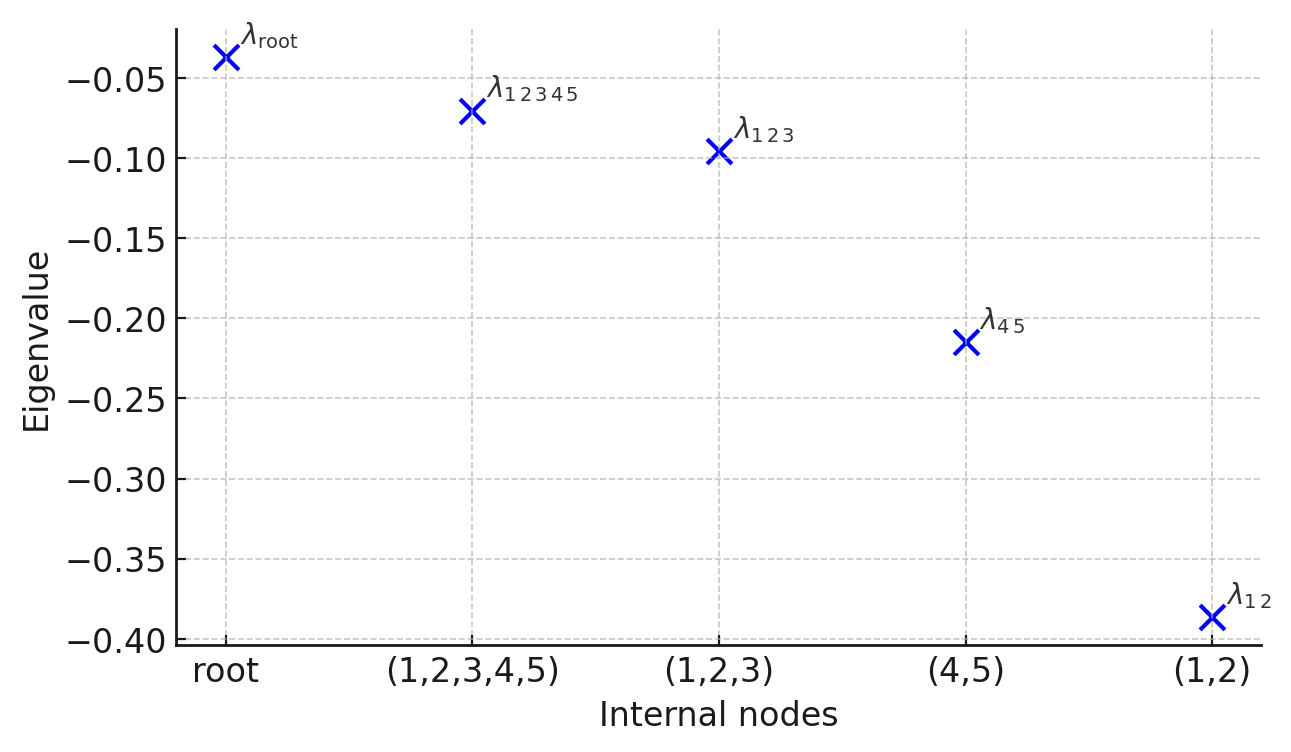}
        \caption{Eigenvalues of a phylogenetic tree}
        \label{fig:phylovalues}
    \end{figure}
    
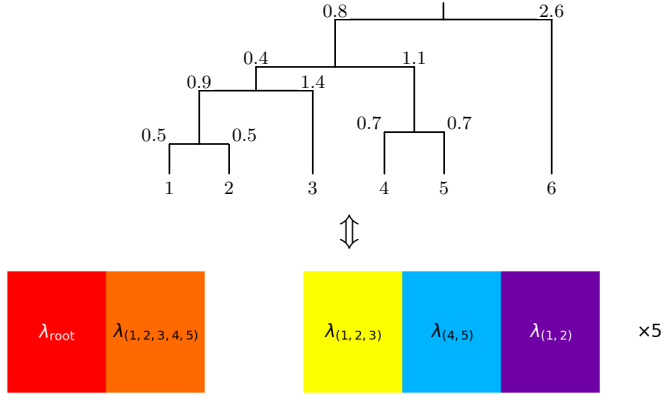
\begin{figure}[H]
  \centering
  \resizebox{0.8\linewidth}{!}{
    \begin{minipage}{\linewidth}
      \centering
      
      \begin{tikzpicture}[
        thick, line cap=round,
        leaf lab/.style={font=\small},
        num/.style={fill=white, inner sep=1pt, font=\scriptsize}
      ]

      \def\Hbase{0.0}
      \def\HoneTwo{0.5}      
      \def\HfourFive{0.7}    
      \def\HoneTwoThree{1.4} 
      \def\Hall{1.8}         
      \def\Htop{2.6}        
      \def\xone{0.6}
      \def\xtwo{1.6}
      \def\xthree{3.0}
      \def\xfour{4.2}
      \def\xfive{5.2}
      \def\xsix{7.0}

      \draw (\xone,\Hbase)   -- (\xone,\HoneTwo);
      \draw (\xtwo,\Hbase)   -- (\xtwo,\HoneTwo);
      \draw (\xthree,\Hbase) -- (\xthree,\HoneTwoThree);
      \draw (\xfour,\Hbase)  -- (\xfour,\HfourFive);
      \draw (\xfive,\Hbase)  -- (\xfive,\HfourFive);
      \draw (\xsix,\Hbase)   -- (\xsix,\Htop);

      \draw (\xone,\HoneTwo) -- (\xtwo,\HoneTwo);                
      \draw (\xfour,\HfourFive) -- (\xfive,\HfourFive);          
      \coordinate (m12) at ($(\xone,\HoneTwo)!0.5!(\xtwo,\HoneTwo)$);
      \coordinate (m45) at ($(\xfour,\HfourFive)!0.5!(\xfive,\HfourFive)$);
      \draw (m12) -- ($(m12)+(0,\HoneTwoThree-\HoneTwo)$);
      \draw (m45) -- ($(m45)+(0,\Hall-\HfourFive)$);

      \draw ($(m12)+(0,\HoneTwoThree-\HoneTwo)$) -- (\xthree,\HoneTwoThree);
      \coordinate (m123) at ($($(m12)+(0,\HoneTwoThree-\HoneTwo)$)!0.5!(\xthree,\HoneTwoThree)$);

      \draw (m123) -- ($(m123)+(0,\Hall-\HoneTwoThree)$);

      \draw ($(m123)+(0,\Hall-\HoneTwoThree)$) -- ($(m45)+(0,\Hall-\HfourFive)$);
      \coordinate (mall) at ($($(m123)+(0,\Hall-\HoneTwoThree)$)!0.5!($(m45)+(0,\Hall-\HfourFive)$)$);

      \draw (mall) -- ($(mall)+(0,\Htop-\Hall)$);
      \draw ($(mall)+(0,\Htop-\Hall)$) -- (\xsix,\Htop);

      \coordinate (midTop) at ($($(mall)+(0,\Htop-\Hall)$)!0.5!(\xsix,\Htop)$);
      \draw (midTop) -- ++(0,0.28);

      \node[num, above left]  at (\xone,\HoneTwo) {0.5};
      \node[num, above right] at (\xtwo,\HoneTwo) {0.5};
     
      \node[num, above left]  at (\xfour,\HfourFive) {0.7};
      \node[num, above right] at (\xfive,\HfourFive) {0.7};
      
      \node[num, above] at ($(m12)+(0,\HoneTwoThree-\HoneTwo)$) {0.9}; 
      \node[num, above] at (\xthree,\HoneTwoThree) {1.4};             
      \node[num, above] at ($(m123)+(0,\Hall-\HoneTwoThree)$) {0.4};    
      \node[num, above] at ($(m45)+(0,\Hall-\HfourFive)$) {1.1};        
      \node[num, above] at ($(mall)+(0,\Htop-\Hall)$) {0.8};            
      \node[num, above] at (\xsix,\Htop) {2.6};                        
      \node[leaf lab, below] at (\xone,\Hbase) {1};
      \node[leaf lab, below] at (\xtwo,\Hbase) {2};
      \node[leaf lab, below] at (\xthree,\Hbase) {3};
      \node[leaf lab, below] at (\xfour,\Hbase) {4};
      \node[leaf lab, below] at (\xfive,\Hbase) {5};
      \node[leaf lab, below] at (\xsix,\Hbase) {6};

      \end{tikzpicture}

      \vspace{0.7em}
      {\Large$\Updownarrow$}
      \vspace{0.7em}

      \includegraphics[width=.9\linewidth]{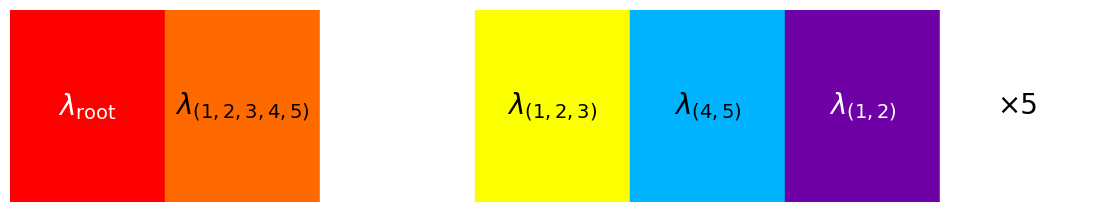}
    \end{minipage}
  }
  \caption{Phylogenetic tree and spectral encoding. For illustration, each eigenvalue was mapped to a corresponding frequency within the visible spectrum. The notation \(\times5\) symbolize \(5\) leafs.}
  \label{fig:tree-to-colors}
\end{figure}
    \end{example}

The extended spectrum is an efficient tool to reconstruct the ultrametric space via the spectrum of \(L_X\) by means of a Breadth-First Search (BFS), which is known to have complexity \(O(|X|)\). It follows that for a given finite ultrametric space, the construction of the spectrum and the measure has also linear complexity. Therefore, the extended spectrum already serves as an optimal tool for the storage, access, and simulation of an ultrametric space.  \newline

Lewitus and Morlon \cite{lewitus2016} introduced a spectral framework for  phylogenetic trees based on the \emph{Modified Graph Laplacian} (MGL), defined as  $\Delta = D - W$, where $D_{ii} = \sum_j w_{ij}$ is the degree matrix and $W$ is  the full pairwise distance matrix between all $2n-1$ nodes of the tree. The eigenvalues of $\Delta$ are used to construct a spectral density profile encoding  global properties of the tree shape. To identify modes of diversification within a tree, they employ the  \emph{eigen-gap heuristic}: if the largest gap in the ranked spectrum falls between $\mu_i$ and $\mu_{i+1}$, the tree is said to contain $i$ clusters of distinct evolutionary dynamics \cite{vonLuxburg2007}. This criterion is explicitly presented as a heuristic, with no formal proof \cite{lewitus2016}.

Our framework differs in three key aspects. First, our operator $L$ acts on the leaves of the tree and carries a natural probabilistic interpretation as the generator of a continuous time Markov chain with jump rates $q(x,y) = k(d(x,y))\,m(y)$, whereas the MGL has no such stochastic interpretation. Second, the eigenvalues of $L$ admit explicit closed-form expressions, making the spectral structure fully transparent without numerical diagonalization. 

Thirdly, in our framework for a particular kernel \(k\) and measure \(m\), the eigen-gaps have a rigorous geometric interpretation: a gap between $\lambda_i$ and $\lambda_{i+1}$ corresponds to a level $h^*$ in the ultrametric hierarchy where the accumulated mass-weighted branch length undergoes a significant jump, providing a formal interpretation in the context of phylogenetic trees.

We now elaborate on the last observation. For a phylogenetic ultrametric tree \(\mathcal{T}\), fix a reference height \(h_0>h(X)\), where \(h(X)\) is the height of the root and define the kernel 
\[k(d(x,y))=F(h_0-h(x\wedge y))\]
Together with the normalized counting measure \(m(x)=1/|X|\) for all \(x\in X\), we define the \emph{ultrametric phylogenetic Laplacian} of \(\mathcal{T}\) with kernel \(F\) as the attached ultrametric Laplacian with this measure and kernel:
\[L_{\mathcal{T}}f(x)=\sum_{y\in X}
F(h_0-h(x\wedge y))(u(y)-u(x))m(y).\]
Let us assume first that \(F(x)=x\). In this case, the jump-rates are \(h_0-h(x\wedge y)\), as discussed in the introduction, the biological interpretation is as follows : two species separate by a very old common ancestor (high height of their LCA) have a small rate jump, that is, jump between phylogenetically distant taxa is a rare event. On the other hand, sister taxa (small height of their LCA) have large jump rates. By Proposition \ref{prop:eigenvalueofultrametricoperator}, the eigenvalues of this operator are
\[\lambda(u)=\sum_{n\in\gamma_r(u)}m(n)l(n,n^+),\]
where \(l(n,n^+)\) is the branch length connecting \(n\) with its immediate ancestor \(n^+\). Hence each eigenvalue accumulates, along the ancestral path from the clade \(n\) to the root, the branch lengths \(l(n,n^+)\) weighted by the mass \(m(n)\) of each ancestral clade, where \(m(n)\) reflects the relative species richness of \(n\). Large eigenvalues correspond to a possible combination of two factors: deeper branches (from the root to the clade) along the path and rich clades with high taxon diversity.  Now let us analyze the gaps. First, let us assume that the clade \(v\) is ancestor of the clade \(u\). In this case
\[\lambda(u)-\lambda(v)=\sum_{n\in \gamma(u,v)} m(n)l(n,n^+).\]
Notice that within a sub-tree characterized by short internal branches (like a radiation event) and more or less homogeneous mass splits, the eigenvalues tend to be similar and collapse to a limiting value as we go down the tree. In this case a large gap would represent a strong asymmetry in the mass distribution or a large branch separating two events. Nevertheless, the deeper the clades, the less likely this becomes. \newline

In the case when \(u\) and \(v\) belong to different lineages we have 
\[\lambda(u)-\lambda(v)=\sum_{n\in \gamma(u,u\wedge v)} m(n)l(n,n^+)-\sum_{n\in \gamma(v,u \wedge v)} m(n)l(n,n^+).\]
Therefore, a strong eigenvalue gap may reflect a deep divergence event, characterized by possibly both, a long internal branch and/or a highly asymmetric mass distribution. From these observations, it is expected that two clusters with different rates of diversification will produce big gaps. As an example we compute the eigenvalue distribution of the phylogenetic ultrametric Laplacian attached to the phylogenetic tree of Primate genera with 109 leaves, obtained from Timetree \cite{Kumar2022}. 

\begin{figure}[H]
    \centering
    \includegraphics[width=0.9\linewidth]{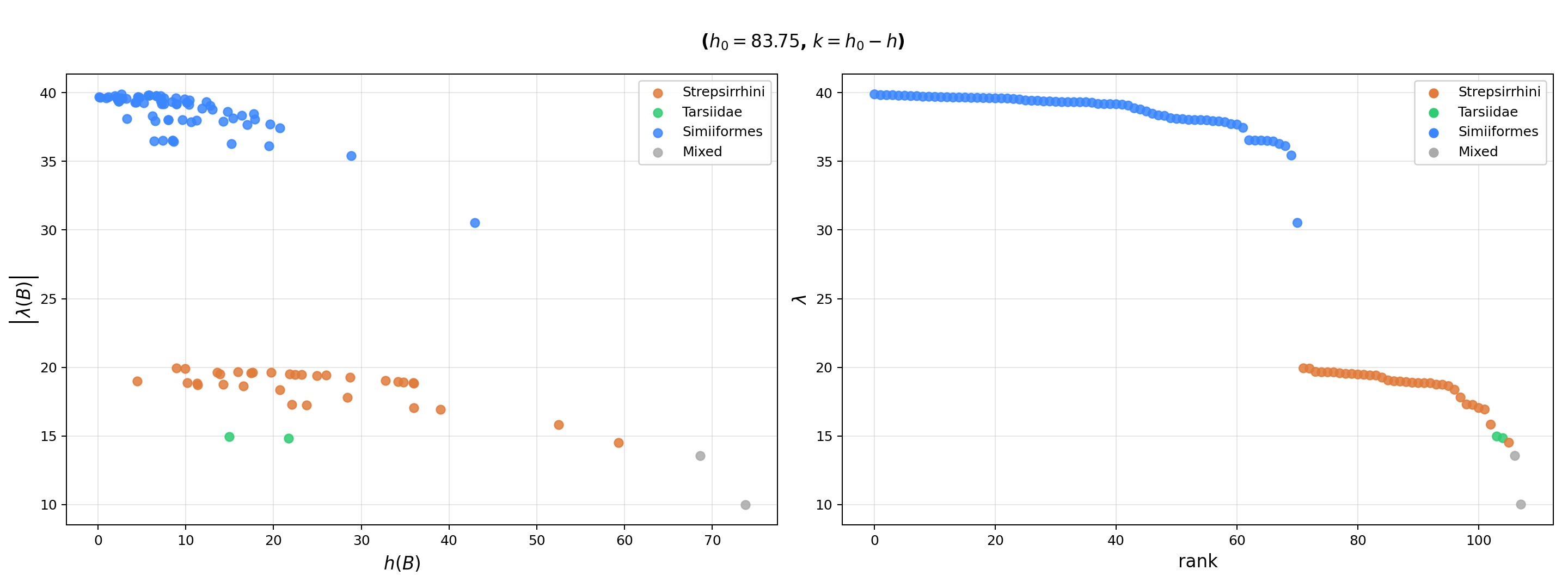}
    \caption{Eigenvalue distribution. Each eigenvalue has attached an internal node. The first plot (left) show the magnitude of each eigenvalue as function of the height. The second one (right) shows the eigenvalue distribution. A noticeable gap is shown in both plots. Each internal point has a color representing it belongs to the class Simiiformes, Tarsiidae (infraorders) or Strepsirrhini (order)   }
    \label{fig:eigenvaluesgapphylo}
\end{figure}
As shown in Figure \ref{fig:eigenvaluesgapphylo}, the gap separates successfully different evolutionary modes. First, left plot show how similar are eigenvalues inside an order (and infra-order) as expected from the theory,  several similar eigenvalues for different heights but different classes of clades. Moreover, we see from this plot that the Simiiformes infra-order have attached higher frequency eigenvalues which suggests a higher diversification in recent times and a higher mass in this clade, whereas the lower values of nodes in the Strepsirrhini order and Tarsiidae infra-order reflect older divergencies and smaller masses as we can see in Figure \ref{fig:primatesgenus}. In the right plot of Figure \ref{fig:eigenvaluesgapphylo}, we see how the eigen-gap separates the two most dominant splits of the phylogenetic tree: Strepsirrhini vs Simiiformes. 

\begin{figure}[H]
    \centering
    \includegraphics[width=0.8\linewidth]{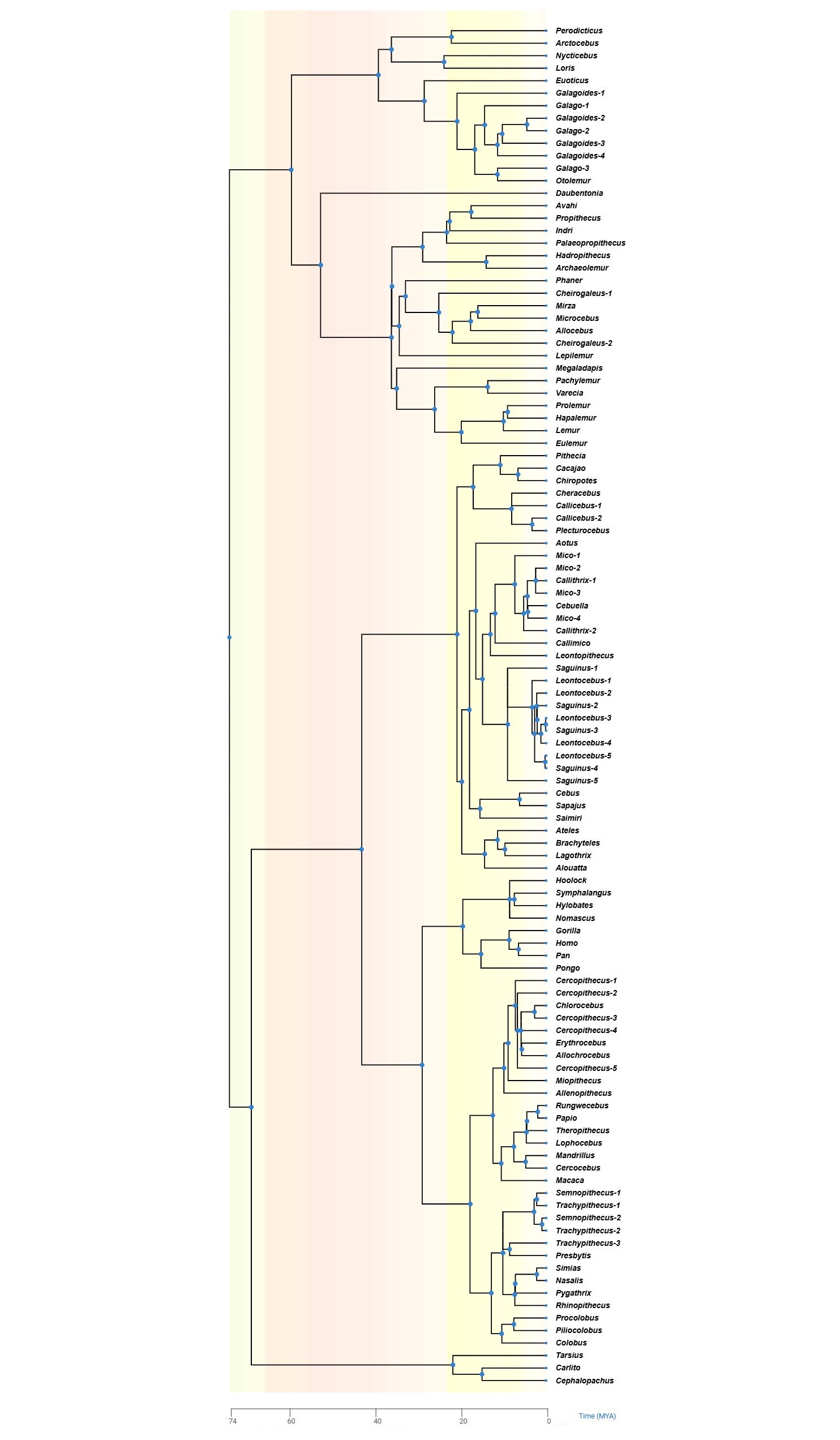}
    \caption{Phylogenetic tree of Primate genera from Timetree.}
    \label{fig:primatesgenus}
\end{figure}
The choice of kernel $F$ in the phylogenetic ultrametric Laplacian is not unique, and different kernels emphasize different scales of the phylogenetic hierarchy. As an illustration, 
consider the sigmoid kernel

\begin{equation}
    F(x) = \frac{1}{1+e^{-\beta(x-c)}},
\end{equation}

with $\beta = 0.3$ and $c = 35$. This kernel acts as a local contrast function, saturating for nodes far from the threshold $c$ and amplifying differences in the neighbourhood of $h_0 - h \approx c$. Figure \ref{fig:spectrum_sigmoid} shows that under this kernel, the eigenvalues associated with the two parvorders of Simiiformes, Platyrrhini and Catarrhini, become completely separated, with no overlap between the two spectral bands. This separation is not as evident with the 
kernel $k(x) = x$. The systematic study of how to select an optimal kernel for a given phylogenetic question, whether to maximize separation at a particular taxonomic scale, or to recover a target spectral structure, is a natural direction for future research.

\begin{figure}
    \centering
    \includegraphics[width=0.5\linewidth]{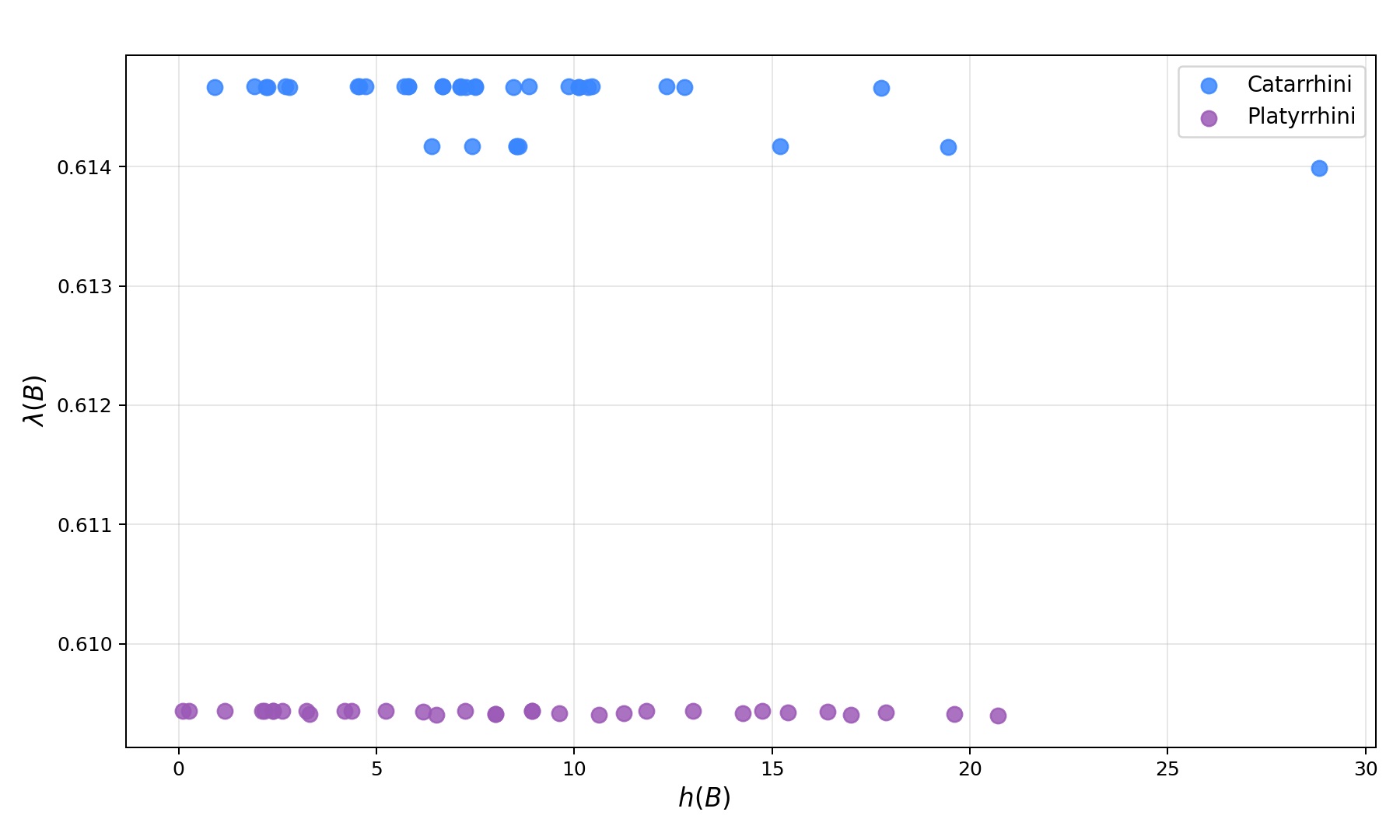}
    \caption{Gap obtained using a sigmoid kernel. }
    \label{fig:spectrum_sigmoid}
\end{figure}

\subsection{The eigenprojectors}
While the spectrum "hears" the geometry of the metric, the eigenprojectors will encode the underlying topology of the ultrametric space, in short, the topological tree can be reconstructed from the projectors/eigenvectors. Let us start by fixing an orthonormal basis \(\mathcal{B}_n\) of \(\mathcal{V}_n\) and denote by 
\[\mathcal{B}_X:=\bigcup_{n\in T\setminus X}\mathcal{B}_n\cup\{\psi_0\},\]
where \(\psi_0\equiv1\) is the trivial eigenvector of a given ultrametric Laplacian \(L_X\). Thus, the set \(\mathcal{B}_X\) is a fixed eigenvector basis of \(L_X\). As we already proved, the support of any eigenvector \(\psi\in \mathcal{V}_n\) is equal to the node \(n\). That is, for an ultrametric phylogenetic tree, each clade has an eigenspace associated with it. We can construct an explicit eigenbasis using the Gram-Schmidt process on the functions \(\varphi_{B_n,l}\). 
\begin{prop}
\label{prop:haarbasistrees}
     Fix an internal node $P$ with disjoint children $C_1,\dots,C_k$. Define
\[
m_r := m(C_r), \qquad s_j := \sum_{r=1}^j m_r.
\]

For $j=1,\dots,k-1$, we define
\[
\psi_{P,j}(x)=
\begin{cases}
\displaystyle
a_j := \sqrt{\dfrac{m_{j+1}}{s_j\,s_{j+1}}},
& x\in C_1\cup\cdots\cup C_j, \\[10pt]

\displaystyle
b_j := -\sqrt{\dfrac{s_j}{m_{j+1}\,s_{j+1}}},
& x\in C_{j+1}, \\[10pt]

0,
& x\in C_r \text{ with } r\ge j+2 \text{ or } x\notin P.
\end{cases}
\]
 Then the set of all functions \(\psi_{P,j}\) spans the space \(\mathcal{V}_n\).
 \end{prop}

\begin{proof}
For each $j=1,\dots,k-1$, the function $\psi_{P,j}$ is supported in $P$ and is constant on each child $C_r$ of $P$. It remains to check that it has zero mean on $P$. Since $\psi_{P,j}$ takes the value $a_j$ on $C_1\cup\cdots\cup C_j$, the value $b_j$ on $C_{j+1}$, and vanishes on the remaining children, we get
\[
\sum_{y\in P} \psi_{P,j}(y)m(y)
=
a_j \sum_{r=1}^j m_r + b_j\,m_{j+1}
=
a_j s_j + b_j m_{j+1}.
\]
Using the definitions of $a_j$ and $b_j$,
\[
a_j s_j
=
s_j\sqrt{\frac{m_{j+1}}{s_j s_{j+1}}}
=
\sqrt{\frac{s_j m_{j+1}}{s_{j+1}}},
\]
and
\[
b_j m_{j+1}
=
-\,m_{j+1}\sqrt{\frac{s_j}{m_{j+1}s_{j+1}}}
=
-\sqrt{\frac{s_j m_{j+1}}{s_{j+1}}}.
\]
Hence
\[
\sum_{y\in P} \psi_{P,j}(y)m(y)=0,
\]
therefore $\psi_{P,j}\in \mathcal V_n$.

Next we prove orthogonality. Let $1\le i<j\le k-1$. Then $\psi_{P,i}$ is constant on $C_1\cup\cdots\cup C_i$, constant on $C_{i+1}$, and zero on $C_r$ for $r\ge i+2$. Since $j\ge i+1$, the function $\psi_{P,j}$ takes the constant value $a_j$ on every child $C_1,\dots,C_{i+1}$. Therefore,
\[
\langle \psi_{P,i},\psi_{P,j}\rangle
=
a_j\sum_{y\in P} \psi_{P,i}(y)m(y)
=
0,
\]
because $\sum_{y\in P} \psi_{P,i}(y)m(y)=0$. Thus the family $\{\psi_{P,j}\}_{j=1}^{k-1}$ is orthogonal.

We now compute the norm of each $\psi_{P,j}$:
\[
\|\psi_{P,j}\|_{L^2(m)}^2
=
a_j^2 \sum_{r=1}^j m_r + b_j^2 m_{j+1}
=
a_j^2 s_j + b_j^2 m_{j+1}.
\]
Substituting the values of $a_j$ and $b_j$ gives
\[
a_j^2 s_j
=
\frac{m_{j+1}}{s_j s_{j+1}}\, s_j
=
\frac{m_{j+1}}{s_{j+1}},
\]
and
\[
b_j^2 m_{j+1}
=
\frac{s_j}{m_{j+1}s_{j+1}}\, m_{j+1}
=
\frac{s_j}{s_{j+1}}.
\]
Hence
\[
\|\psi_{P,j}\|_{L^2(m)}^2
=
\frac{m_{j+1}}{s_{j+1}}+\frac{s_j}{s_{j+1}}
=
\frac{s_{j+1}}{s_{j+1}}
=
1.
\]
So the family is orthonormal.

Finally, any function in $\mathcal V_n$ is determined by its constant values on the $k$ children $C_1,\dots,C_k$, subject to the single linear relation
\[
\sum_{r=1}^k c_r\,m_r=0.
\]
Therefore,
\[
\dim \mathcal V_n = k-1.
\]
Since $\{\psi_{P,j}\}_{j=1}^{k-1}$ is an orthonormal family of $k-1$ elements contained in $\mathcal V_n$, it is an orthonormal basis of $\mathcal V_n$. In particular, it spans $\mathcal V_n$.
\end{proof}
This basis is a generalization of the basis presented in \cite{Gorman2023} which is recovered in the case of a binary tree with the counting measure. In that work, the eigenbasis is introduced as an adaptation for a more general framework for wavelets in trees. Nevertheless, here they appear naturally as eigenvectors of the ultrametric Laplacian. In \cite{Gorman2023} this basis is used for the sparsification of huge ultrametric matrices and therefore, is proposed as a method for storing big ultrametric trees like the tree of life. \newline

We now specialize to the binary case. We assume that \(m\) is a multiple of the counting measure and the tree is binary. Let \(f:X\rightarrow \mathbb{R}\) be a function defined on the leaves of the tree, which in biological applications corresponds to a trait measured across species. Then this function can be decomposed in eigen-modes via the ultrametric basis derived in Proposition \ref{prop:haarbasistrees}. 
\[f=f_0+\sum_{P}c_P\psi_P,\]
where \(f_0=\sum_{x\in X}f(x)m(x)\) and 
\[c_P=\langle f,\psi_P\rangle.\]
The \(c_P\) are closely related with the variance of \(f\) with respect to the measure \(m\):
\[Var_m(f):=||f-\bar f||^2=\sum_{x}m(x)(f(x)-\bar f)^2=\sum_{P}c_P^2.\]
And, explicitly
\[c_P=\sqrt{\frac{m(C_1)m(C_2)}{m(P)}}(\bar f_{C_1}-\bar f_{C_2}),\]
where 
\[\bar f_{C_i}=\frac{1}{m(C_i)}\sum_{x\in C_i}m(x)f(x),\]
for \(i\in \{1,2\}\). 

Therefore, each eigenmode is proportional to the contrast associated with the divergence \(P\), that is, the difference between the averages of the traits in the split generated by the sub clades \(C_1\) and \(C_2\). The proportionality factor is \(\sqrt{\frac{m(C_1)m(C_2)}{m(P)}}\) which penalizes asymmetries in the split. Substituting into the variance identity we obtain
\[Var_m(f)=\sum_{P}\frac{m(C_1)m(C_2)}{m(P)}(\bar f_{C_1}-\bar f_{C_2})^2.\]
Moreover, each summand can be expressed as
\[c_P^2=m(C_1)(\bar f_{C_1}-\bar f_P)^2+m(C_2)(\bar f_{C_2}-\bar f_P)^2\]
where 
\[\bar f_P=\frac{m(C_1)\bar f_{C_1}+m(C_2)\bar f_{C_2}}{m(P)},\]
hence \(c_P^2\) measures the between-group variance, i.e. how much the two groups differ with respect to the average in the clade \(P\), hence the total variance \(Var_m(f)\) is decomposed into variances between the splits generated by the phylogenetic tree.  \newline

Thus, we expect the largest contributions to \(Var_m(f)\) to arise from splits where substantial differences of the trait occur between two clades of comparable mass. This suggests using the coefficients \(c_P\) as a natural framework for phylogenetic comparison, where differences between traits are analyzed through the orthogonal contrasts associated with the divergence events of the tree. As an example, let us analyze three traits on the phylogenetic tree of Primate genera. The information for these traits was obtained from the PanTHERIA dataset \cite{Pantheria}.

\begin{figure}[H]
    \centering
    \includegraphics[width=0.75\linewidth]{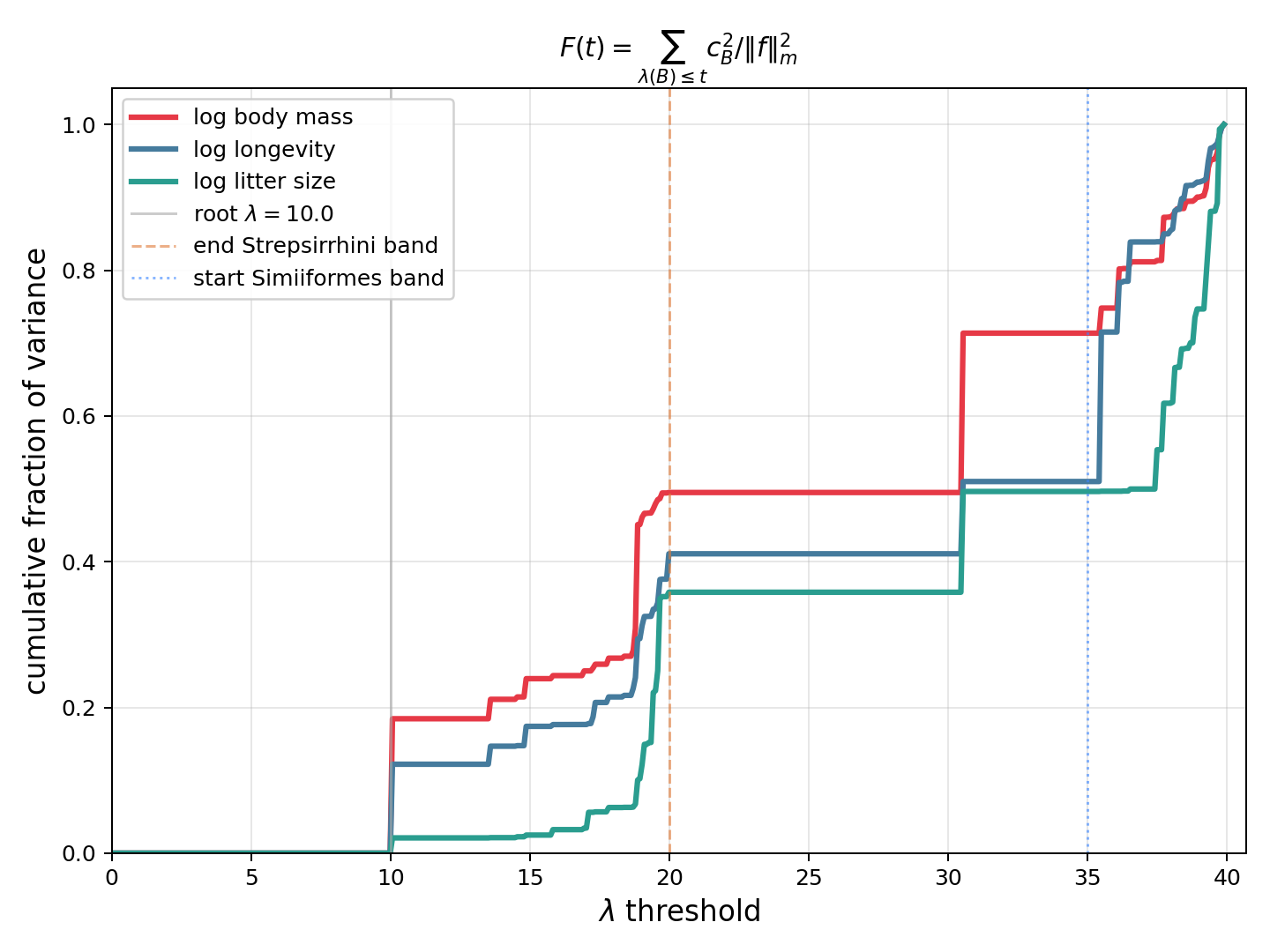}
    \caption{Phylogenetic cumulative variance function of the traits: body mass, longevity and litter size.}
    \label{fig:spectral_variance}
\end{figure} 

Figure \ref{fig:spectral_variance} illustrates the phylogenetic cumulative fraction of variance
of three life-history traits measured across primate genera: log body mass, log
maximum longevity, and log litter size. For each internal node $B$, the squared
Fourier coefficient $c_B^2 = \langle f, \psi_B \rangle_m^2$ quantifies
the fraction of total trait variance explained by the contrast between the two
child clades of $B$. The cumulative function

\begin{equation}
    F(t) = \frac{\sum_{\lambda(B) \leq t} c_B^2}{\|f\|_m^2}
\end{equation}

aggregates these contributions by eigenvalue.

The three traits exhibit different profiles. For log body mass approximately $18.5\%$ of its variance is explained by the root node alone
($\lambda \approx 10$), reflecting the large difference in mean body size
between Strepsirrhini and Haplorhini. A secondary contribution arises within
the Strepsirrhini band ($\lambda \in [18, 20]$), accounting for a further
$22.7\%$ of variance through internal contrasts. However, the
dominant contributor to body mass variance is the node separating Platyrrhini
from Catarrhini ($\lambda \approx 30.5$, $h \approx 43$ Ma), which alone
explains $21.9\%$ of total variance. This split distinguishes the New World
monkeys from the Old World monkeys and apes.
The contrast between the mean body mass of these two
assemblages, weighted by the mass distribution $m$, is the single largest
source of body mass variation across the primate tree. By $\lambda = 35$,
approximately $71\%$ of body mass variance is accounted for.

Log maximum longevity presents a qualitatively different pattern, with $49\%$
of its variance concentrated in the Simiiformes band ($\lambda \in [35, 40]$).
The dominant node is the split between Hominoidea and Cercopithecidae
($\lambda = 35.4$, $h \approx 29$ Ma), which alone explains $20.5\%$ of total
longevity variance, reflecting the substantially greater maximum lifespan of
great apes relative to Old World monkeys. The root contributes only $12.2\%$,
indicating that the Strepsirrhini--Haplorhini divergence is less informative
for longevity than for body mass.

Log litter size exhibits the flattest cumulative profile. 
The root contributes only $2.1\%$ of variance, meaning
that clade membership at the deepest level is nearly uninformative about litter
size. Instead, variance is spread across multiple bands, with the largest single
contributor being the Platyrrhini--Catarrhini split ($\lambda \approx 30.5$,
$13.8\%$). Overall, litter size variance is distributed more uniformly across eigenvalue bands than either body mass or longevity.

Taken together, the three spectral profiles illustrate that the ultrametric
Laplacian decomposition recovers biologically meaningful structure: traits
with strong ancient signal (body mass) accumulate variance rapidly at low
$\lambda$, while traits shaped by more recent or lineage-specific evolution
(longevity, litter size) show flatter profiles with variance concentrated at
higher eigenvalues.

\subsection{The heat kernel of an ultrametric Laplacian} \label{chapter2}
In this section we study the heat kernel associated to an ultrametric Laplacian. We approach this concept from the perspective of the system of differential equations naturally generated by the Laplacian matrix, that is, from the point of view of a master equation or a continuous-time Markov chain (CTMC), therefore we are interested in the following initial value problem or \emph{Cauchy problem} :

Let $(X,d)$ be a finite ultrametric space, and let $L$ and $L_X$ denote the associated ultrametric Laplacian matrix and operator, respectively. For a function $u : X \to \mathbb{R}$, the \emph{evolution} \(u(\cdot,t)\) is governed by
\[
\frac{d}{dt} u(t,x) = (L_X u(t,\cdot))(x),
\]
with initial condition $u(0,x) = u_0(x)$. Equivalently, in matrix form this can be written as
\[
\frac{d}{dt} u(t) = L u(t), \qquad u(0) = u_0.
\]

By a straightforward computation, one can obtain the general solution of the Cauchy problem by using the spectral structure of the operator $L_X$. In particular, let $\{\lambda_n\}$ denote the spectrum of $L_X$ and let $\{\psi_n\}$ be an eigenbasis of $L^2(X,m)$ as introduced previously. Expanding the initial condition in this basis, we write
\[
u_0 = \sum_{n} \langle u_0, \psi_n \rangle_{L^2(X,\mu)} \, \psi_n.
\]
The above equation is also called \emph{The heat equation} or \emph{The diffusion equation}. Substituting this expansion into the evolution equation yields the explicit solution
\[
u(t) = \sum_{n} e^{\lambda_n t} \, \langle u_0, \psi_n \rangle_{L^2(X,m)} \, \psi_n.
\]
Equivalently, for the initial data $u(0) = u_0$, one has
\[
u(t) = e^{tL} u_0,
\]
where $\{e^{tL}\}_{t \geq 0}$ denotes the semigroup generated by the ultrametric Laplacian matrix $L$. Equivalently, in terms of the operator $L_X$ this reads
\[
u(x,t) = e^{tL_X} u_0(x).
\]
This semigroup formulation naturally leads to the notion of the heat kernel associated with \(L_X\).

\begin{definition}
Let $(X,d,m)$ be a measurable finite ultrametric space and $L_X$ the corresponding ultrametric Laplacian operator. The function
\[
H_t : (0,\infty) \times X \times X \longrightarrow \mathbb{R},
\]
satisfying that for every function $u : X \to \mathbb{R}$ and every $t > 0$,
\[
(e^{tL_X} u)(x) = \sum_{y\in X} H_t(x,y)\, u(y)\, m(y),
\]
is called the \emph{heat kernel} of $L_X$.
\end{definition}
From the definition it follows that the heat kernel is completely determined by the spectrum of $L_X$ and its associated orthonormal eigenbasis. Indeed, if $\{\lambda_n\}$ denotes the spectrum of $L_X$ and $\{\psi_n\}$ an orthonormal eigenbasis of $L^2(X,m)$, then
\begin{equation}\label{heatkernelgen}
  H_t(x,y) = \sum_{n} e^{\lambda_n t} \, \psi_n(x) \, \overline{\psi_n(y)}.
\end{equation}
Every eigenvalue \(\lambda_n\) is associated with an internal node \(n\in T\). Therefore, we can write the contribution of this eigenvalue and the corresponding eigenfunctions \(\mathcal{B}_n\) as follows: 
\[H^{(n)}_t(x,y):=e^{\lambda_n t} \sum_{\psi\in \mathcal{B}_n}\psi(x)\psi(y).\]
Let \(E_n(x,y)\) the projection kernel to the eigen-space \(\mathcal{V}_n\) , that is, for each \(f\in L^2(m)\),
\[\sum_{y\in X} E_n(x,y)f(y)m(y)-f(x) \, \perp \, \mathcal{V}_n.\]
We have
\[E_n(x,y)=\sum_{\psi\in \mathcal{B}_n}\psi(x)\psi(y). \]
In particular for \(\delta_z\) where \(z\in B_j\subset B_n\) and \(j\in C(n)\), 
\[E(x,z)m(z)=\sum_{y\in X} E_n(x,y)\delta_z(y)m(y)=m(z)\left(\frac{\textbf{1}_{B_j}}{m(B_j)}(x)-\frac{\textbf{1}_{B_{n}}}{m(B_{n})}(x)\right)=m(z)\varphi_{B_j}(x),\]
where the last equality follows from the fact that \(m(z)\varphi_{B_j}-\delta_z \, \perp \, \varphi_{B_l}\), for all \(l\in C(n)\). Consequently 
\begin{equation}\label{projectionkernel}
    E_n(x,z)=\sum_{\psi\in \mathcal{B}_n}\psi(x)\psi(z)=\varphi_{z\in B_j}(x).
\end{equation}
We now show that the matrix $H_t = (H_t(x,y))_{x,y \in X}$ inherits a block structure by the topological tree of $X$. More precisely, each internal node of the topological tree corresponds to a block of $H_t$, and the entries within each block are constant for fixed \(t>0\).

\begin{thm}
Let $(X,d,m)$ be a measurable finite ultrametric space with associated topological tree $T$ and heat kernel $H_t = (H_t(x,y))_{x,y \in X}$. Then for each internal node $n \in T$ the subset of indices corresponding to the children of $n$ determines a block of the matrix $H_t$. Moreover, each such block has constant entries, that is, for a given \(t>0\),
\begin{equation} \label{eq:heatkernelteorem}
    H_t(x,y)=1+e^{\lambda(B_n) t }\left(\frac{\delta_{xy}}{m(\{x\})}-\frac{1}{m(B_n)}\right)+\sum_{T:B_n\subseteq T \subsetneq X} e^{\lambda(T^{+}) t}\left(\frac{1}{m(T)}-\frac{1}{m(T^{+})}\right),
\end{equation}
where \(B_n=  x\wedge y \).
\end{thm}

\begin{proof}
Let \(n\in T\) be an internal node. Let \(x,y\in X\) such that \(x\wedge y=n\). Therefore, there exists two disjoint balls \(B_i,B_j\) with \(i,j\in C(n)\) such that \(x\in B_i\) and \(y\in B_j\). By equation \ref{projectionkernel}:
\[E_n(x,y)=\varphi_{y\in B_j}(x)=-\frac{1}{m(B_n)}.\]
On the other hand, let \(j\in \gamma_r(n)\setminus\{r\}\), then \(x,y\in B_j\), therefore
\[E(x,y)_{F(j)}=\varphi_{y\in B_j}(x)=\left(\frac{1}{m(B_j)}-\frac{1}{m(B_{F(j)})}\right).\]
By summing along the history of \(n\) and including the contribution of the trivial eigenvector, we obtain the desired equality.
\end{proof}
From the semigroup structure we have the following asymptotic formula, for \(x,y\in X\),
\[
H_t(x,y)m(y)=\delta_{xy}+t\,L_{xy}+\frac{t^2}{2}\,(L_{xy}^2)(x,y)+O(t^3)\quad (t\downarrow 0).
\]
If \(x\in B_l\), \(y\in B_m\) with \(l\neq m\) and \(n=x\wedge y\), by expanding the exponential functions in expansion \ref{eq:heatkernelteorem} and comparing it with the above equation, we can write the matrix \(k(d(x,y))\) in terms of the eigenvalues and the measure: 
\begin{equation} \label{kerneleigenvalueequation}
    k(d(x,y))=-\frac{\lambda(B_n)}{m(B_n)}+\sum_{T:B_n\subseteq T \subsetneq X } \left(\frac{1}{m(T)}- \frac{1}{m(T^+)}\right)\lambda(T^{+}).
\end{equation}
We now focus on the long-time behavior of the heat kernel. Let \(\pi:X\times X \rightarrow \mathbb{R}\) defined as \(\pi(x,y)\equiv 1\). This kernel has attached an operator of the form \[\Pi u(x)=\sum_{y\in X} \pi(x,y)u(y)m(y)=\sum_{y\in X}u(y)m(y).\]
Therefore the heat kernel can be write as
\begin{equation*}
  H_t(x,y)=\pi(x,y)+\sum_{n} e^{\lambda_n t} \psi_n(x)\psi_n(y),
\end{equation*}
where \(\psi_n\) is an eigenvector attached to the internal node \(n\). Let \(f\in L^2(X,\mu)\), with the expansion \(f=f_0 1_X(x)+\sum f_n \psi_n \). Therefore,
\[||(e^{tL_X}-\Pi)f||_{L^2}^2=\sum_n e^{2\lambda_n t} f_n^2 \leq e^{2\lambda_{gap}t}||f||_{L^2},\] where \(\lambda_{gap}\) is the eigenvalue with minimal absolute value. Therefore, \(||e^{tL_X}-\Pi||_{L^2}\leq e^{\lambda_{gap} t}\). Moreover, if \(\psi_n\) is a Kozyrev wavelet attached to \(\lambda_{gap}\), then \(||(e^{tL_X}-\Pi)\psi_n||_{L^2}=e^{\lambda_{gap} t}\), hence
\begin{equation}\label{L2normheatkernel}
  ||(e^{tL_X}-\Pi)||_{L^2}=e^{\lambda_{gap}t}. 
\end{equation}
Therefore, the decaying rate of the heat kernel is controlled by the \emph{spectral gap} \(|\lambda_{gap}|\). 

Moreover, since  for every function $u : X \to \mathbb{R}$ and every $t > 0$,
\[
(e^{tL_X} u)(x) = \sum_{y\in X} H_t(x,y)\, u(y)\, m(y),
\]
we have that 
\begin{equation}
    \label{ergodicitylimitequation}
    \lim_{t\rightarrow\infty}e^{tL_X}u(x)=\sum_{y\in X}u(y)m(y).
\end{equation}
Although the matrix of \(L_X\) respect the canonical basis is not symmetric unless \(m\) is uniform, \(L_X\) is a self-adjoint operator in \(L^2(X,m)\). Equivalently, the ultrametric Laplacian matrix \(L\) satisfies the \emph{detailed balance condition}, namely 
\[m(x)L_{xy}=m(y)L_{yx}.\]
The self adjoint-ness of \(L_X\) and equation \ref{ergodicitylimitequation} implies the following identity.
\begin{equation}
    \label{eq:invariantmeasure} 
    \sum_{y\in X}e^{tL_X}u(y)m(y)=\sum_{y\in X}u(y)m(y), \quad\forall \ t>0
\end{equation}
Indeed, by taking the derivative respect time of the LHS, we obtain
\[\sum_{y\in X}L_X(e^{tL_X}u(y))m(y)=\langle L_X(e^{tL_X}u),\textbf{1}\rangle_{L^2(m)}= \langle e^{tL_X}u,L_X\textbf{1}\rangle_{L^2(m)}=0,\]
where last equality holds since \(\textbf{1}\) is the trivial eigenvector of \(L_X\) with eigenvalue zero. This identities will have a clear probabilistic meaning in the next section. Since the rate of convergence of the semigroup to its limit for \(t\rightarrow \infty\) depend on the spectral gap that is the first non zero eigenvalue we close this section with the following proposition which characterize the spectral gap.

\begin{proposition}
  If \(k\) is decreasing,  
\[|\lambda_{gap}|=k(diam(X)).\]
\end{proposition}
\begin{proof}
  Let \(n\in T\) be an internal node. Let \(m\in T\) a proper descendant. Then \(X\setminus 
  B_n\subset X\setminus 
  B_m\). Let \(x_0\in B_m\), then
  \[\sum_{X\setminus B_m} k(d(x_0,y))d\mu(y)=\sum_{X\setminus B_n} k(d(x_0,y))d\mu(y)+k(\diam(B_n))(\mu(B_)).\]
Therefore by equation \ref{eigenvaluereal} we obtain
\[\lambda_m-\lambda_n=\mu(B_m)(k(\diam(B_n))-k(\diam(B_u))).\]
If \(k\) is a decreasing function and since \(\diam(B_m)\leq \diam(B_n)\), then \(\lambda_m \leq \lambda_n<0\). As a consequence, \(\lambda_{gap}=\lambda_r\), where \(r\in T\) is the root. 
\end{proof}

\section{Dynamic centrality and ultrametric spaces.} \label{sec:centrality}
\subsection{A state-centrality index for CTMC.}
Several complex systems can be described through stochastic dynamics evolving on a large configuration space.
Such systems typically explore a vast number of states before approaching equilibrium. 
In statistical physics, this behavior is often interpreted through the notion of an underlying energy landscape, where each configuration is associated with a potential energy and the dynamics describe random transitions or jumps between metastable states. Under suitable assumptions, the resulting dynamics satisfy the Markov property and can be modeled as a CTMC \cite{Peliti2021Stochastic}.

While this framework arises naturally in non-equilibrium statistical mechanics, similar Markovian descriptions appear in broader contexts where the state space exhibits an intrinsic hierarchical organization. 
In such situations, the geometry of the configuration space plays a fundamental role in shaping the stochastic evolution. 
Rather than focusing on a specific physical model, we concentrate here on the structural properties of ultrametric state spaces and investigate how continuous-time Markov dynamics reflect their hierarchical organization.

In particular, this perspective motivates the study of geometric quantities associated with the generator of the dynamics, such as state-centrality indices, which quantify the structural relevance of individual states independently of a specific physical interpretation.
Therefore the continuous time stochastic process $X_t$, where $t>0$ and $X_t\in S$, where $S$ is the finite state space (the space of configurations), is assumed to satisfy the Markov property (and therefore generates a  CTMC) $\mathbb{P}(X_{t}:X_{t_{n}})=\mathbb{P}(X_{t}:X_{t_{n}},X_{t_{n-1}},...,X_{t_{0}})$. During the evolution of the dynamics, given any two states $i,j\in S$, we are interested in describing the transition probabilities $p_{i,j}(t)$, describing the probability of jumping to state $j$ at time $t$, given that $X_0=i$. Using the Markov property and the law of total probability, it is possible to describe this probability function in terms of the so called master equation
\begin{equation} \label{mastereq}
    \frac{d}{dt}P(t)=QP(t),
\end{equation}
where $p_{ij}(t)=(P(t))_{ij}$ and $q_{ij} = (Q)_{ij}=:\frac{d}{dt} p_{ij}(t) |_{t=0}$. The numbers $q_{ij}$ are called the transition rates of the Markov process. The transition rates $q_{ij}$ determine the infinitesimal behavior of the conditional probabilities by $\mathbb{P}(X_{t+h}=j:X_{t}=i)=\delta_{ij} + q_{ij}h+o(h)$. \newline

The general solution of equation \ref{mastereq}, is given by the attached semigroup $P(t)=\exp(tQ)$. One important conclusion from this description is that when the state space $S$ is finite, the CTMC is determined by its transition matrix $Q$. Therefore, is usually convenient to describe the process via its Transition Network (or Jump Network), see \cite{Peliti2021Stochastic}. 

This network is constituted by the state space \(S\) as the set of vertices and the directed edges are weighted by the corresponding rated $q_{ij}$, whenever this are positive, no directed edge is given when the rate transition is zero. We say that a transition network is connected if any state $j$ can be reached from any other state $i$. This assumption is usually made in transition rate phenomena (chemical reactions, biomolecules, etc) since disconnected master equations represent multiple non interacting physical systems. From this always follow the relaxation of the system toward an equilibrium state, that is, a stationary distribution exists. We denote this stationary distribution by $\pi=(\pi(i))_{i\in S} $. Therefore we have $\lim_{t\rightarrow \infty}p_{ij}(t)=\pi_j$. In many thermodynamic systems, it is common to assume a stronger condition called the detailed balance condition given by the relation
\begin{equation} \label{detailedbalance}
    \pi(i)q_{ij}=q_{ji} \pi(j),
\end{equation}
it can be shown that condition \ref{detailedbalance} holds if and only if 
\begin{equation}
\label{balanceprob}
    \pi(i)p_{ij}(t)=p_{ji}(t)\pi(j).
\end{equation}
for every $t\geq0$. 

\subsection*{Random walk centrality}
In \cite{Nohcentrality2004}, Noh and Rieger introduce the concept of random walk centrality, which measures the capacity of a node of a network to receive and redistribute information. Next, we briefly review the definition of this index based on a discrete-time random walk in a network setting. \newline

Let \( G = (V,E) \) be a finite, connected, undirected graph with adjacency matrix \(A\). Denote by
\[
k_i = \sum_{j \in V} A_{ij}
\]
the degree of node \(i\).

A discrete-time random walk on \(G\) is a Markov chain \((X_t)_{t \ge 0}\) with state space \(V\) and transition probabilities
\[
P_{ij} = \mathbb{P}(X_{t+1}=j \mid X_t=i) = \frac{A_{ij}}{k_i}.
\]
That is, at each time step, a walker located at node \(i\) chooses uniformly at random one of its neighbors and moves to it.

Since the graph is connected and undirected, the Markov chain admits a unique stationary distribution given by
\[
\pi_i = \frac{k_i}{\sum_{j \in V} k_j}.
\]
The characteristic relaxation time of node \(i\) is defined as
\[
\tau_i = \sum_{t=0}^{\infty} \big( P_{ii}(t) - \pi_i \big),
\]
where \(P_{ii}(t)\) denotes the probability that a walker starting at \(i\) is at \(i\) at time \(t\).

The \emph{random walk centrality} of node \(i\) is then defined as
\[
C_i = \frac{\pi_i}{\tau_i}.
\]

The quotient \(C_i\) captures the balance between two effects: the visitation probability of the node, encoded in the stationary weight \(\pi_i\), and the local time relaxation behavior of the walk around that node, encoded in \(\tau_i\). Hence, \(C_i\) quantifies how efficiently node \(i\) can receive and redistribute information under a random walk dynamics on the network.

Moreover, the mean first passage times satisfy the relation
\[
\mathbb{E}[T_{j} \mid X_0 = i ] - \mathbb{E}[T_{i}\mid X_0 = j ] = C_j^{-1} - C_i^{-1},
\]
where
\[
T_{j} := \inf \{ t \ge 0 : X_t = j \}
\]
denote the \emph{first passage time} (or hitting time) from to node \(j\), that is, the first time at which a random walker starting at \(i\) visits \(j\). This shows that nodes with larger random walk centrality are, on average, reached more rapidly by the random walk.

\subsection{Dynamic centrality}
We now extend the result of Noh and Rieger  to the time-continuous case. To the best of our knowledge, this result have not appeared explicitly in previous literature.

In a similar way, for a \(CTMC\) with probability transitions \(p_{ij}(t)\) we define
\[
T_{ij} := \inf \{ t \ge 0 : X_t = j \}
\]
denote the first passage (hitting) time of node \(j\). The mean first passage time (MFPT) from \(i\) to \(j\) is defined as
\[
m_{ij} := \mathbb{E}_i[T_{ij}],
\]
where \(\mathbb{E}_i[\cdot]\) denotes expectation conditioned on \(X_0 = i\). Following \cite{vanKampen2007} the following equation holds.  
\begin{equation}
    \label{renewal}
    \hat{f}_{i,j}(s)=\frac{\hat{p}_{i,j}(s)}{\hat{p}_{j,j}(s)},
\end{equation}
where $\hat{f}_{ij}(s)=\int_{0}^{\infty}e^{-st}f_{i,j}(t)dt,$ is the Laplace transform of the first-passage probability, and $\hat{p}_{ij}(s)=\int_{0}^{\infty}e^{-st}p_{ij}(t)dt$ is the Laplace transform of the transition probability $p_{ij}(t)$.  \newline

Using equation \ref{renewal} we obtain 
\[m_{ij}=-\left.\frac{d}{ds}\hat{f}_{ij}(s)\right|_{s=0}
 \]

Define $R_{ij}^{(m)}:=\int_{0}^{\infty} t^m (p_{ij}(t)-\pi(j))dt$. For $s>0$, in virtue of the dominated convergence theorem and the uniform convergence of the series expansion of the function  $x\mapsto e^{-x}$ on compacts, we have
\begin{equation*}
    \begin{split}
        S\left(R_{ij}^{(m)}\right)(s)=\sum_{m=0}^{\infty} R_{ij}^{(m)}(-1)^{m}\frac{s^m}{m!}
         &=\int_{0}^{\infty}\left(\sum_{m=0}^{\infty}t^m \frac{s^m}{m!}(-1)^m\right)(p_{ij}(t)-\pi(j))dt \\
         &=\int_{0}^{\infty}e^{-st}(p_{ij}(t)-\pi(j))dt \\
         &=\hat{p}_{ij}(s)-\frac{\pi(j)}{s}.
    \end{split}
\end{equation*}
Therefore, 
\[\hat{f}_{ij}(s)=\frac{\hat{p}_{ij}(s)}{\hat{p}_{jj}(s)}=\frac{\pi(i)+sS\left(R_{ij}^{(m)}\right)(s)}{\pi(i)+S\left(R_{jj}^{(m)}\right)(s)}.\]
For $i\neq j$ the above equality lead to 

\begin{equation}
    \label{meanpassage}
    m_{ij}=-\left.\frac{d}{ds}\hat{f}_{ij}(s)\right|_{s=0}=\frac{R_{jj}^{(0)}}{\pi(j)}-\frac{R_{ij}^{(0)}}{\pi(j)}.
\end{equation}
Note that, equation \ref{balanceprob}, implies that $\pi(j)R_{ji}^{(0)}=\pi(i)R_{ij}^{(0)}$, hence

\[m_{ij}-m_{ji}=\left( \frac{R_{ji}^{(0)}}{\pi(i)}-\frac{R_{ij}^{(0)}}{\pi(j)} \right)+\left( \frac{R_{jj}^{(0)}}{\pi(j)}-\frac{R_{ii}^{(0)}}{\pi(i)} \right)=\frac{R_{jj}^{(0)}}{\pi(j)}-\frac{R_{ii}^{(0)}}{\pi(i)},\]
\begin{definition}
    For a state \(i\in S\) we define its CTMC centrality as the number 
    \[C_{CTMC}(i)=\frac{\pi(i)}{R_{ii}^{(0)}}.\]
\end{definition}
Hence in the CTMC case, we can make the same conclusion as in \cite{Nohcentrality2004}; the CTMC centrality $C_{CTMC}(i)$ quantifies how central is the state $i$ regarding its potential to be accessible from other states. That is, the following implication holds for two states $i,j\in S$: 

\begin{equation*}
    C_{CTMC}(i)>C_{CTMC}(j)\iff m_{ij}>m_{ji}.
\end{equation*}

Therefore, in average, the system access from the state $j$ to the state $i$ faster than from $i$ to the state $j$. This is the continuous time analog to the result and conclusion obtained by Noh and Rieger in \cite{Nohcentrality2004}. 

\subsection{Continuous time Markov Chain in ultrametric spaces}
The matrix representation of the operator \(L_X\) is a \(Q\)-matrix , and therefore, the operator generates a stochastic process where the transition probability matrix is given by \(P_t=exp(tL_X)\). The main properties of this process are described in the Theorem below. 

\begin{thm}
    Let \((X,d,m)\) be a finite measurable ultrametric space with ultrametric Laplacian \(L_X\) and \(m\) being a probability measure. Then \(L_X\) generates an irreducible continuous time Markov chain with transition function \[P_t(x,y)=e^{tL_X}\delta_{y}(x)=H_t(x,y)m(y).\]
The process is reversible, and the measure \(m\) is the only stationary probability measure.
\end{thm}
\begin{proof}
    Since \(L\) is a \(Q\)-matrix, it has attached a continuous time Markov chain, and since \(L_{xy}>0\) for all \(x,y\in S\), \(x\neq y\), then the Markov chain is irreducible. The equality from the Theorem follows from the relation \(e^{tL_X}\delta_y(x)=(e^{tL})_{xy}\), for all \(x,y\in S\). The fact that \(m\) is a stationary measure for the process follows from equation \ref{eq:invariantmeasure}.
    \end{proof}
From this result and the last section, we now have a clear probabilistic interpretation of the off-diagonal entries of \(L\):
\[k(d(x,y))m(y)=\frac{d}{dt}\mathbb{P}(X_{t+h}=y|X_t=x) |_{t=0},\]
that is, \(k(d(x,y))m(y)\) is the instantaneous transition rate, or in other words, the transition density rate per unit of time from \(x\) to \(y\).  Since \(d\) is ultrametric and \(k\) is a non-increasing function, the process is compatible with the topology, in the following sense: if \(d(x,y)\leq d(x,z)\) and \(m(y)=m(z)\) then \(\mathbb{P}(X_{t+h}=y|X_t=x)\geq\mathbb{P}(X_{t+h}=z|X_t=x)\), i.e., the stochastic system at \(X_t=x\) has more probability to occupy the nearest states. It is clear that the measure \(m\) can bias the jumping rate, nevertheless in the next applications \(m\) will be the normalized counting measure, hence the property \(m(y)=m(z)\) is satisfied trivially for all pair of points. 

On the other hand, notice that for the ultrametric phylogenetic Laplacian \(L_{\mathcal{T}}\) the following holds: 
\[
F(h_0-h(x\wedge y))=\frac{d}{dt}\mathbb{P}(X_{t+h}=y|X_t=x)\big|_{t=0},
\]
hence, the random process attached to this generator is compatible with the phylogenetic structure of the tree, in particular for \(F(x)=x\) the rates depend linearly on the divergence time.

\subsection{Dynamic centrality for Ultrametric spaces as a topological descriptor}

We now proceed to study the CTMC centrality on ultrametric spaces. Denote by \(\tau_i:=R_{ii}^{(0)}\). Let \((X,d)\) a finite ultrametric space. And let \(P_t(x,y)\) be the probability transition function attached to the operator \(L_X\) with kernel \(k\). In order to study the centrality \(C_{CTMC}(i)\), we need to investigate the quantities \(\tau_i\) and \(\pi(i)\), for the latter we have the following result.

\begin{lem}
    Let \((X,d)\) a finite ultrametric space equipped with a probability measure \(m(x)\), and \(P_t(x,y)\) the probability transition function attached to the operator \(L_X\).  Then \(\pi(i)=m(i)\) for all \(i\in X\), that is 
    \(\lim_{t\rightarrow \infty}P_t(x,y)=m(y),\quad x,y\in X.\)
\end{lem}
\begin{proof}
    The results follows from equation \ref{L2normheatkernel}, since \(P_t(x,y)=H_t(x,y)m(y)\). 
\end{proof}
This allows us to simplify the expression of \(\tau_i\), 
\[\tau_i=\int_{0}^{\infty}(P_t(i,i)-\pi(i))dt=m(i)\int_{0}^{\infty}(H_t(i,i)-1)dt\]
therefore,
\[C_{CTMC}(i)=\frac{1}{\int_{0}^{\infty}(H_t(i,i)-1)dt}\]
Recall that 
\[H_t(x,x)-1=e^{\lambda(B_n) t }\left(\frac{1}{m(\{x\})}-\frac{1}{m(B_n)}\right)+\sum_{T:B_n\subseteq T \subsetneq X} e^{\lambda(T^{+}) t}\left(\frac{1}{m(T)}-\frac{1}{m(T^{+})}\right),\]
therefore, 
\begin{equation*}
    \begin{split}
        \int_{0}^{\infty}(H_t(i,i)-1)dt&=\\
        &-\frac{1}{\lambda(B_n)}\left(\frac{1}{m(\{x\})}-\frac{1}{m(B_n)}\right)+\sum_{T:B_n\subseteq T \subsetneq X} -\frac{1}{\lambda(T^{+})}\left(\frac{1}{m(T)}-\frac{1}{m(T^{+})}\right)
    \end{split}
\end{equation*}
where \(n=F(i)\), the father of node \(i\), this leads to the following result for finite ultrametric spaces. 

\begin{thm}
    Given a finite ultrametric space \((X,d)\) with probability measure \(m\), the CTMC centrality attached to the CTMC generated by \(L_X\) is given by 
    \begin{equation}
\label{centralityultrametricequation} 
C_{CTMC}(i)=\left(-\frac{1}{\lambda(B_n)}\left(\frac{1}{m(\{i\})}-\frac{1}{m(B_n)}\right)+\sum_{T:B_n\subseteq T \subsetneq X} -\frac{1}{\lambda(T^{+})}\left(\frac{1}{m(T)}-\frac{1}{m(T^{+})}\right)\right)^{-1}
    \end{equation}
\end{thm}
It is clear that the topology of the ultrametric space will affect the accessibility of the sates during the random process. This effect can be effectively studied using the equation \ref{centralityultrametricequation}. Since \(k\) is non increasing, it follows that the eigenvalues as a function of one of the radius of a given ball is a decreasing function, therefore, increasing the radius of one of the balls, say \(B_n\), will decrease the index \(C_{CTMC}(i)\) for all states \(i\in X\) such that \(n\in \gamma_{\raiz}(i)\), hence one state will be more accessible when the balls in its history \(\gamma_{\raiz}\) have smaller radii. \newline

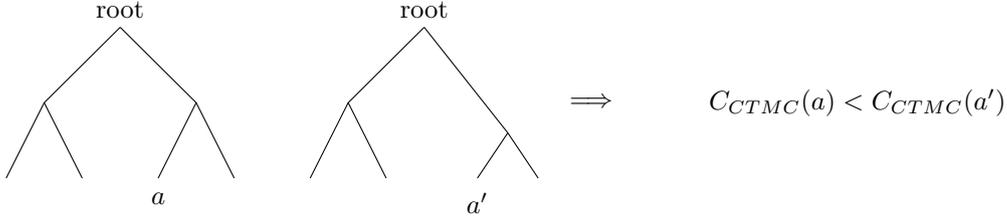
\begin{figure}[h!]
    \centering
    \begin{tikzpicture}[scale=1]

  \coordinate (r1) at (0,0);
  \coordinate (l1) at (-1,-1);
  \coordinate (r1r) at (1,-1);

  \coordinate (a) at (-1.5,-2);
  \coordinate (b) at (-0.5,-2);
  \coordinate (c) at (0.5,-2);
  \coordinate (d) at (1.5,-2);

  \draw (r1) -- (l1) -- (a);
  \draw (l1) -- (b);
  \draw (r1) -- (r1r) -- (c);
  \draw (r1r) -- (d);

  \node[above] at (r1) {root};
  \node[below=2pt] at (a) {};
  \node[below=2pt] at (b) {};
  \node[below=2pt] at (c) {$a$};
  \node[below=2pt] at (d) {};

  \coordinate (r2) at (4,0);
  \coordinate (l2) at (3,-1);
  \coordinate (ap) at (2.5,-2);
  \coordinate (bp) at (3.5,-2);

  \coordinate (r2r) at (5.1,-1.4);
  \coordinate (cp) at (4.7,-2);
  \coordinate (dp) at (5.5,-2);

  \draw (r2) -- (l2) -- (ap);
  \draw (l2) -- (bp);
  \draw (r2) -- (r2r);
  \draw (r2r) -- (cp);
  \draw (r2r) -- (dp);

  \node[above] at (r2) {root};
  \node[below=2pt] at (ap) {};
  \node[below=2pt] at (bp) {};
  \node[below=2pt] at (cp) {$a'$};
  \node[below=2pt] at (dp) {};

  \node at (6.2,-1) {$\Longrightarrow$};

  \node at (9.7,-1) {$C_{CTMC}(a) < C_{CTMC}(a')$};

\end{tikzpicture}
    \caption{Decreasing the radius of a ball make the states inside of it more accessible.}
\end{figure}
The other parameter which may affect the centrality is the measure of a given ball. When \(m(B_n)\) increases, then the eigenvalue absolute value of the affected eigenvalues increases, we conclude that from equation \ref{centralityultrametricequation}, that increasing the measure of a ball increases the centrality of its elements. 

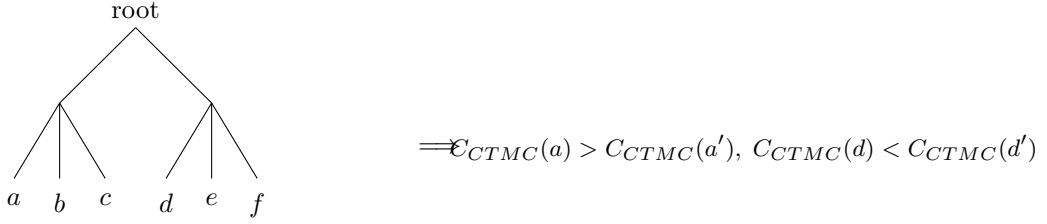
\begin{figure}[h!]
    \centering
\begin{tikzpicture}[scale=1]

  \begin{scope}
    \coordinate (r1) at (0,0);
    \coordinate (l1) at (-1,-1);
    \coordinate (r1r) at (1,-1);

    \coordinate (a) at (-1.6,-2);
    \coordinate (b) at (-1.0,-2);
    \coordinate (c) at (-0.4,-2);

    \coordinate (d) at (0.4,-2);
    \coordinate (e) at (1.0,-2);
    \coordinate (f) at (1.6,-2);

    \draw (r1) -- (l1);
    \draw (r1) -- (r1r);

    \draw (l1) -- (a);
    \draw (l1) -- (b);
    \draw (l1) -- (c);

    \draw (r1r) -- (d);
    \draw (r1r) -- (e);
    \draw (r1r) -- (f);

    \node[above] at (r1) {root};
    \node[below=2pt] at (a) {$a$};
    \node[below=2pt] at (b) {$b$};
    \node[below=2pt] at (c) {$c$};
    \node[below=2pt] at (d) {$d$};
    \node[below=2pt] at (e) {$e$};
    \node[below=2pt] at (f) {$f$};
  \end{scope}

  \begin{scope}[yshift=-3.2cm]
    \coordinate (r2) at (0,0);
    \coordinate (l2) at (-1,-1);   
    \coordinate (r2r) at (1,-1);

    \coordinate (ap) at (-1.4,-2);
    \coordinate (bp) at (-0.6,-2);

    \coordinate (cp) at (0.2,-2);
    \coordinate (dp) at (0.8,-2);
    \coordinate (ep) at (1.4,-2);
    \coordinate (fp) at (2.0,-2);

    \draw (r2) -- (l2);
    \draw (r2) -- (r2r);

    \draw (l2) -- (ap);
    \draw (l2) -- (bp);

    \draw (r2r) -- (cp);
    \draw (r2r) -- (dp);
    \draw (r2r) -- (ep);
    \draw (r2r) -- (fp);

    \node[above] at (r2) {root};
    \node[below=2pt] at (ap) {$a'$};
    \node[below=2pt] at (bp) {$b'$};
    \node[below=2pt] at (cp) {$c'$};
    \node[below=2pt] at (dp) {$d'$};
    \node[below=2pt] at (ep) {$e'$};
    \node[below=2pt] at (fp) {$f'$};
  \end{scope}

  \node at (4,-1.6) {$\Longrightarrow$};

  \node at (8,-1.6) {\scriptsize
    $C_{CTMC}(a) > C_{CTMC}(a'),\; C_{CTMC}(d) < C_{CTMC}(d')$
  };

\end{tikzpicture}

    \caption{For the counting measure, increasing the number of leaves of an internal node increases their accessibility.}
\end{figure}
Since, \begin{equation*}
    C_{CTMC}(i)>C_{CTMC}(j)\iff m_{ij}>m_{ji}.
\end{equation*}
We can make explicit the inequality of the RHS. According to the previous section, for a general  CTMC, \(m_{ij}=\frac{R_{jj}^{(0)}}{\pi(j)}-\frac{R_{ij}^{(0)}}{\pi(j)},\) therefore
\begin{equation}
\label{eq:MFPTCTMC}
    \begin{split}
        m_{ij}&=\int_{0}^{\infty}(H_t(j,j)-H_t(i,j))dt\\
        &=-\frac{1}{\lambda(B_n)}\cdot\frac{1}{m(B_n)}-\frac{1}{\lambda(B_{F(j)})}\left(\frac{1}{m(j)}-\frac{1}{m(B_{F(j)})}\right)+\sum_{T:B_{F(j)}\subseteq T\subsetneq B_n}-\frac{1}{\lambda(T^+)}\left(\frac{1}{m(T)}-\frac{1}{m(T^+)}\right)\\
        &=-\frac{1}{\lambda(B_F(j))}\cdot\frac{1}{m(j)}+\sum_{T:B_{F(j)\subseteq T\subsetneq B_n}}\frac{1}{m(T)}\left(\left(-\frac{1}{\lambda(T^+)}\right)-\left(-\frac{1}{\lambda(T)}\right)\right)
    \end{split}
\end{equation}
As we already established,  \(m_{ij}\) is the average time that, for the first time the system occupy the state \(j\) given that the process started at \(i\). Hence, dynamically, the accessibility to the leaf \(j\) from the node \(i\) depend only on its history up to the node \(n=LCA(i,j)\). Topologically,  the quantity \(m_{ij}\) capture the topological ramification of this path. \newline

Indeed, first, the behavior of \(m_{ij}\) depend on the behavior of the differences between the radius of two consecutive balls: If \(\Delta r=|r(a)-r(b)|\), for \(a,b\in\gamma_{n}(j)\) increases, then \(m_{ij}\) increases as well due the functional dependence of \(  \left(-\frac{1}{\lambda(T^*)}\right)-\left(-\frac{1}{\lambda(T^*)}\right)\) as a function of \(\Delta r\). In other words, once the system enters the ball/cluster \(B_n\) it finds easier to access the state \(j\) if the radii of the nested ball containing \(j\) are smaller. Secondly, if the measure of one of those balls increases, then the system has more possible accessible states in the cluster \(B_n\), thereby decreasing the time \(m_{ij}\). This  shows that \(m_{ij}\) reflects the ramification and topology of the path \(j\rightarrow LCA(i,j)\): smaller \(m_{ij}\) means a more ramified (in the topological tree) path \(j\rightarrow LCA(i,j)\) or more compactly nested balls containing \(j\). \newline

Therefore, \(C(i)>C(j)\) does not only have a dynamic meaning, but also a topological one, the inequality can be interpreted as \textit{"the path \(i\rightarrow LCA(i,j)\) is more ramified or connects more quickly to other leaves that the path \(j\rightarrow LCA(i,j)\) }". A more dynamically isolated state give rise to a more isolated leaf in the ultrametric space. This will be central for many of the applications later on. To give an example on how the indices reflect richness of the topology we have the following corollary. 
\begin{cor}
    If the finite ultrametric space is level-regular then \(C(a)=C(b)\) for all \(a,b\in X\).  
\end{cor}
Let us close this section with an application to phylogenetic trees. As before, we use the phylogenetic tree of Primate genera, we compute the CTMC centrality for two different kernels of the ultrametric phylogenetic Laplacian, \(h_0-h\) and \(1/d\), where \(d=2h\). 

\begin{figure}[H]
    \centering
    \includegraphics[width=1\linewidth]{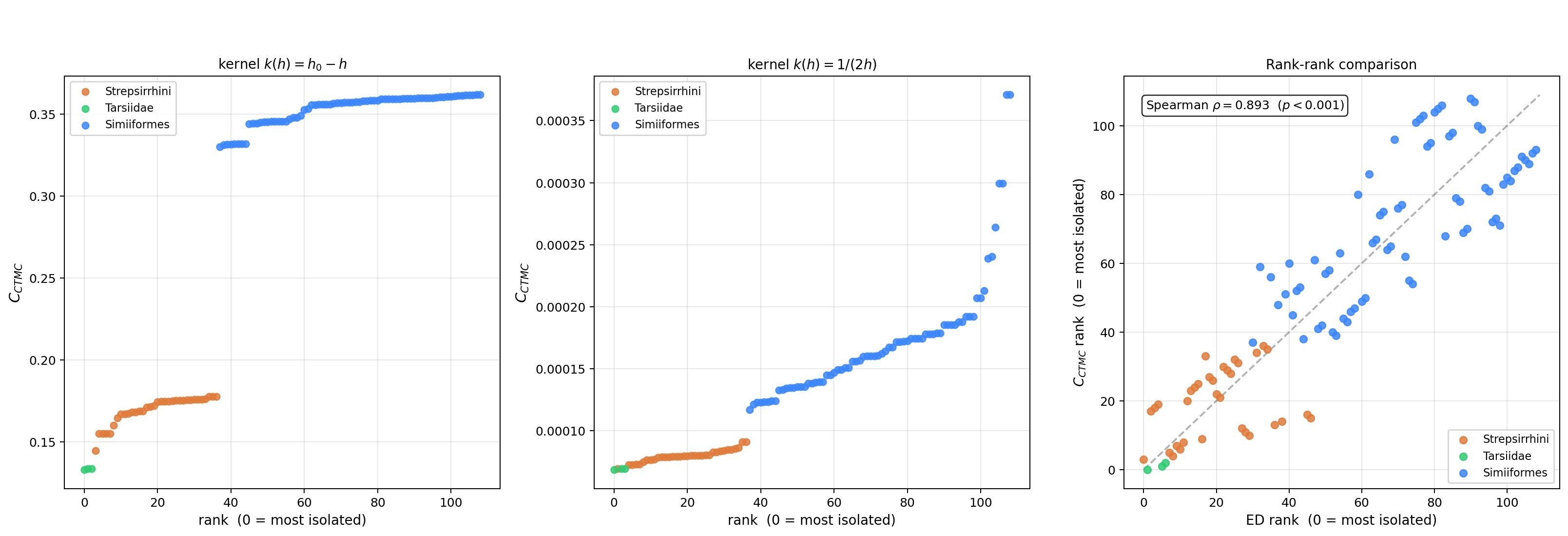}
    \caption{Dynamic centrality for the Primate genus tree. First plot (from left to right) shows the centrality for the kernel \(k=h_0-h\), second for kernel \(k=1/2h\). The last plot show the high but not full correlation between the classic ED rank and the rank given by the \(C_{CTMC}\) index.  }
    \label{fig:ctmc_centrality}
\end{figure}

Figure \ref{fig:ctmc_centrality} shows the CTMC centrality index
$C_{CTMC}(i)$ computed for all $109$ primate genera under two kernel
choices, alongside a rank-rank comparison with the Evolutionary
Distinctiveness (ED) score, see \cite{Redding2014}. In the first two
panels, genera are ordered by ascending centrality, so the leftmost
points correspond to the most phylogenetically isolated species, which under the linear kernel $k(h) = h_0 - h$ (left panel) are: the three
Tarsiidae (\textit{Tarsius}, \textit{Carlito}, \textit{Cephalopachus})
and \textit{Daubentonia}, the only representative of family
Daubentoniidae. Lorisidae
(\textit{Loris}, \textit{Nycticebus}, \textit{Arctocebus},
\textit{Perodicticus}) follow immediately, consistent with their
position as an ancient and species-poor clade. At the opposite extreme,
the most central genera belong to the Cercopithecinae, a dense and
species-rich radiation with many close relatives sharing long internal
branches. Under the kernel $k(h) = 1/(2h)$ (middle panel), which
down-weights ancient splits and up-weights recent divergences, the
global ordering is preserved at the extremes but differs in the middle
ranks, where recently diversified clades
gain centrality relative to the linear kernel.

The principal discrepancies occur for genera such as
\textit{Megaladapis}, \textit{Lepilemur}, and \textit{Phaner}, which
ED ranks among the five most isolated due to their long pendant edges,
while $C_{CTMC}$ assigns them substantially lower isolation ranks
because their parent nodes carry relatively high spectral weight,
reflecting the presence of multiple close relatives within
Strepsirrhini. Conversely, \textit{Daubentonia} is ranked as the most
isolated by ED but only fourth by $C_{CTMC}$, where the three
Tarsiidae displace it at the extreme. This illustrates a fundamental difference between the
two indices: ED captures \emph{local} uniqueness along the path to the
root:
\begin{equation*}
    ED(i) = \frac{1}{N} \sum_{T:\, \{i\} \subseteq T \subseteq X}
    \frac{h(T^+) - h(T)}{m(T)}.
\end{equation*}
In contrast $C_{CTMC}$ incorporates the \emph{global} spectral
structure of the tree, weighting each split by its eigenvalue
$\lambda(B)$, which depends on the mass distribution of the entire
phylogeny.

The CTMC centrality index $C_{CTMC}$ offers several structural differences compared with existing measures of evolutionary isolation reviewed in 
\cite{Redding2014}. First, and most fundamentally, it is not a
heuristic: it emerges directly from the spectral theory of the
ultrametric Laplacian operator $L_{\mathcal{T}}$ and admits a precise
probabilistic interpretation in terms of the dynamics of a
continuous-time random walk on the leaves of the phylogenetic tree. This grounds the
index in a mathematical framework rather than an ad-hoc
scheme. Second, the kernel function $k$ provides an interpretable resolution tool in which to tune the trade-off between
\emph{uniqueness} (sensitivity to recent, tip-near divergences) and
\emph{originality} (sensitivity to ancient, root-near divergences)
identified by \cite{Redding2014} as the key dimension separating
existing metrics; each kernel choice is mathematically justified within
the CTMC framework rather than chosen arbitrarily. Third, because
$C_{CTMC}$ is defined through the eigenvalues of a global operator, it
incorporates information from the entire tree topology rather than only
the path from a species to the root. Finally, the formulation extends
naturally to non-ultrametric trees and phylogenetic networks by
redefining the underlying operator, addressing a limitation explicitly
noted by \cite{Redding2014} for most existing metrics.

\section{Conclusions and outlook.}
We have developed a unified spectral framework for finite ultrametric 
phylogenetic trees, grounding the analysis of phylogenetic structure in 
operator theory and stochastic dynamics. The results presented here, 
spectral reconstruction, eigenmode trait decomposition, and CTMC centrality, are exact, computationally efficient, and biologically interpretable, 
and they are supported by numerical experiments on empirical primate data.

Several directions remain open. The eigenbasis introduced here provides a 
natural interface with Geometric Deep Learning: ultrametric Laplacians admit 
explicit diagonalization, making their spectral parameters directly 
interpretable within graph neural network architectures, and we conjecture 
that this could prove fundamental for the development of phylogenetic 
comparative methods in that framework. A second direction concerns the 
stochastic independence of the contrasts: while the eigenmode decomposition 
is orthogonal by construction, it remains to identify a stochastic process 
under which the coefficients $c_P$ are statistically independent, in the 
spirit of Felsenstein's independent contrasts. Finally, the systematic study 
of kernel selection, which taxonomic scale to emphasize, and how to 
recover a target spectral structure, and a deeper comparison of 
$C_{\mathrm{CTMC}}$ with modern conservation indices beyond ED, represent 
natural next steps.

\section*{Data availability}
The phylogenetic tree of Primate genera used in this study was obtained from the TimeTree database \cite{Kumar2022}. Trait data were obtained from the PanTHERIA dataset \cite{Pantheria}. No new data were generated in this work.

\section*{Acknowledgements}

 Patrick Bradley is warmly thanked for the many useful conversations we had and for his general advice. This work is supported by the Deutsche Forschungsgemeinschaft
under project number 469999674.

\bibliography{biblio}

\end{document}